\documentclass[11pt,a4paper]{article}
\usepackage{graphicx,amssymb,amsmath,amsthm,amsfonts,amsthm,epsf,url,natbib}
                        \usepackage{color,xspace}

\newcommand{\field}[1]{\mathbb{#1}}
 
\newcommand{\R}{\field{R}}

\newcommand{\Pd}{\field{P}} 
\newcommand{\Qd}{\field{Q}}

\newcommand{\kmax}{{k_{\rm max}}}

\newcommand{\dom}{{\rm dom}}

\newcommand{\essinf}{{\rm ess \, inf}}
\newcommand{\ess}{{\rm ess}\!\!\!}
\newcommand{\MaxLoss}{{\rm MR}}

\newcommand{\Pbar}{{\overline{\Pd}}}
\newcommand{\thbar}{{\overline{\theta}}}

\newcommand{\Labar}{{\overline{\Lambda}}}

\newcommand{\thmin}{{\theta_{\rm min}}}

\newcommand{\intdom}{{\mbox{int\;dom\;}}}

\newtheorem{corollary}{Corollary}
\newtheorem{lemma}{Lemma}

\newtheorem{proposition}{Proposition}
\newtheorem{theorem}{Theorem}

\newtheoremstyle{remarkstyle}% name of the style to be used
  {5mm}% measure of space to leave above the theorem. E.g.: 3pt
  {5mm}% measure of space to leave below the theorem. E.g.: 3pt
  {\normalfont}% name of font to use in the body of the theorem
  {0pt}% measure of space to indent
  {\bfseries}% name of head font
  {.}% punctuation between head and body
  {2mm}% space after theorem head
  {}% Manually specify head
\theoremstyle{remarkstyle}

\newtheorem{remark}{Remark}

\begin{document}

\title{Measuring Distribution Model Risk\thanks{Thomas Breuer, PPE Research Centre, FH Vorarlberg, \texttt{thomas.breuer@fhv.at}. Imre Csisz\'ar, Alfr\'ed R\'enyi Institute of Mathematics, Hungarian Academy of Sciences, \texttt{csiszar.imre@mta.renyi.hu}. This work has been supported by the Hungarian National Foundation for Scientific Research under Grant K76088 and by the Austrian Forschungsf\"{o}rderungsgesellschaft in its Josef Ressel Center for Optimisation under Uncertainty. The second
author has benefited from discussions with Franti\u{s}ek Mat\'u\u{s}.
% We gratefully acknowledge comments by Klaus B\"ocker, Freddy Delbaen, Georg Pflug, Chris Rogers, Walter Schachermayer, Martin Summer, as well as to seminar participants in London, Munich, and Vienna.
}
}
%\thanks{Thomas Breuer, . Imre Csisz\'ar, . We gratefully acknowledge comments by Klaus B\"ocker, Freddy Delbaen, Georg Pflug, Chris Rogers, Walter Schachermayer, Martin Summer, as well as to seminar participants in London, Munich, and Vienna.}

\author{Thomas Breuer \and Imre Csisz\'ar}
\date{21 January 2013}
\maketitle

\begin{abstract}
We propose to interpret distribution model risk as sensitivity of expected loss to changes in the risk factor distribution, and to measure the distribution model risk of a portfolio by the maximum expected loss over a set of plausible distributions defined in terms of some divergence from an estimated distribution. The divergence may be relative entropy, a Bregman distance, or an $f$-divergence. We give formulas for the calculation of distribution model risk and explicitly determine the worst case distribution from the set of plausible distributions. We also give formulas for the evaluation of divergence preferences describing ambiguity averse decision makers.
\end{abstract}


\vspace{.3cm} \noindent
% \setlength{\baselineskip}{1.5\baselineskip}

% \vspace{.3cm} \noindent
% {\bf JEL classification:} C18, C44, C60, G01, G32, M48\\
% {\bf AMS classification:} 62C20, 90B50, 91B30, 94A17 
% \setlength{\baselineskip}{1.5\baselineskip}
\newpage

\section{The problem of model risk}
Financial risk measurement, pricing of financial instruments, and portfolio selection are all based on statistical models. If the model is wrong, risk numbers, prices, or optimal portfolios are wrong. Model risk quantifies the consequences of using the wrong models in risk measurement, pricing, or portfolio selection.  

The two main elements of a statistical model in finance are a risk factor distribution and a pricing function. Given a portfolio (or a financial instrument), the first question is: On which kind of random events does the value of the portfolio depend? The answer to this question determines the state space $\Omega$.\footnote{It is possible to choose a larger state space including variables which do not affect the value of the given portfolio. This could allow to compare different portfolios, which do not all depend on the same risk factors. Typically modellers try to keep the number of risk factors small and therefore use a smaller state space. With various techniques they try model some risk factors as a function of a smaller set of risk factors. Thus the number of risk factors actually used in the model, although it may go into the thousands, is typically much smaller than the number of variables influencing the loss.} A point $r\in\Omega$ is specified by a collection of possible values of the risk factors. The state space codifies our lack of knowledge regarding all uncertain events affecting the value of a given instrument or portfolio. The specification of a distribution class and some parameter estimation procedure applied to historical data determines some best guess risk factor distribution, call it $\Pd_0$. The second central element is a pricing function $X: \Omega\rightarrow \R$ describing {\em how} risk factors impact the portfolio value at some given future time horizon. We work in a one-stage set-up. Often modellers try to use only risk factors which are (derived from) prices of basic financial instruments. Describing the price of the portfolio as a function of the prices of these basic instruments is a modelling exercise, which is prone to errors. It involves asset pricing theories of finance with practically non-trivial assumptions on no arbitrage, complete markets, equilibrium, etc. Together the risk factor distribution and the pricing function determine the profit loss distribution. In a last step, a risk measure associates to the profit loss distribution a risk number describing a capital requirement.

Corresponding to the two central elements of a statistical model we distinguish two kinds of model risk: distribution model risk and pricing model risk. This paper is concerned with {\em distribution model risk}.\footnote{\cite{Gibson2000} uses the term model risk for what we call distribution model risk. For a first classification of model risks we refer to \cite{CrouhyGalaiMark1998}. Distribution model risk encompasses both estimation risk and misspecification risk in the sense of \citet{KerhofMelenbergSchumacher2010}, but here we do not need to distinguish the two.} (For an interesting approach to pricing model risk we refer to \citet{Cont2006}.) Although $\Pd_0$ is a best guess of the risk factor distribution, one is usually aware that due to model specification errors or estimation errors the data generating process might be different from $\Pd_0$. Distribution model risk should quantify the consequences of working with $\Pd_0$ instead of the true but unknown data generating process. We propose to measure distribution model risk by 
\begin{equation}
\label{eq-Delbaenrep}
\MaxLoss:= - \inf_{\Pd\in\Gamma} E_\Pd(X)
\end{equation}
where $\Gamma$ is some set of plausible alternative risk factor distributions. So $\MaxLoss$ is the negative of the worst expected value which could result if the risk factor distribution is some unknown distribution in $\Gamma$. %(This $\inf_{\Pd\in\Gamma} E_\Pd(X)$ is denoted by $V$.) 
We propose to choose for $\Gamma$ balls of distributions, defined in terms of some divergence, centered at $\Pd_0$: 
\begin{equation}
\label{eq-Gamma}
\Gamma=\{\Pd: D(\Pd\, || \, {\Pd_0}) \leq k\},
\end{equation}
where the divergence $D$ could be the relative entropy (synonyms: Kullback-Leibler distance, $I$-divergence), some Bregman distance, or some $f$-divergen\-ce. $\Gamma$ contains all risk factor distributions $\Pd$ whose divergence from ${\Pd_0}$ is smaller than some radius $k>0$. The parameter $k$ has to be chosen by hand %, similarly to the choice of a threshold probability $\alpha\in[0,1]$ for Value at Risk or Expected Shortfall
and describes the degree of uncertainty about the risk factor distribution. For larger values of $k$ the set of plausible alternative distributions is larger, which is appropriate for situations in which there is more model uncertainty. 
In Section~\ref{sec-relative entropy} we give the definitions of various divergences and discuss the choice of divergence $D$.

In a previous paper (\cite{BreuerCsiszar2012}) we have addressed 
the problem \eqref{eq-Delbaenrep} for the special case where $D$ is the relative entropy, assuming some regularity conditions which ensure the worst case distribution solving
\eqref{eq-Delbaenrep} is from some exponential family.
 %the problem has been addressed by \citet{BreuerCsiszar2012}. 
The present paper first extends those results, giving the solution for the pathological
cases when these regularity  
conditions are not met (Section~\ref{section-RelEnt-pathological}). 
Second, as main mathematical result, we provide the solution to 
Problem~\eqref{eq-Delbaenrep}, including the characterization
of the minimiser when it exists, for $\Gamma$ of the form \eqref{eq-Gamma} defined in terms of a convex 
integral functional (Section~\ref{sec-CIF}). The special 
cases of Bregman balls and $f$-divergence balls are treated in
Section~\ref{sec-applications}. 
Finally, in Section~\ref{sec-divergence preferences} 
we will address the related, mathematically simpler, problem 
\begin{equation}
\label{DP}
W:= \inf_{\Pd} \left[ E_\Pd(X) + \lambda D(\Pd\, || \, {\Pd_0})\right], \:\:\:\: \lambda > 0.
\end{equation}
Decision makers with divergence preferences rank alternatives $X$ by this 
criterion. 
We apply the methods of Section~\ref{sec-CIF} to derive an explicit solution for
the divergence preference problem~\eqref{DP}.

Mathematically, our approach will be to exploit the relationship of
Problem~\eqref{eq-Delbaenrep} to that of minimizing convex integral 
functionals (and specifically relative entropy) under moment constraints. 
The tools we need do not go beyond convex duality for $\R$ and $\R^2$,
and many results directly follow from known ones about the moment
problem. While in the literature attention is frequently restricted
to essentially bounded $X$, here the $\Pd_0$-integrability of $X$ suffices.

\section{Relation to the literature}
Problem \eqref{eq-Delbaenrep} has been addressed in the literature in
two related contexts: coherent risk measures and ambiguity. Law-invariante risk 
measures assign to a profit loss distribution a number interpreted as risk capital. 
\citet{ADEH1999} and \citet{FollmerSchied2004} formulated requirements for risk measures and coined the terms `coherent' resp. `convex' for risk measures fulfulling them. {\em Every} coherent risk measure can be represented as \eqref{eq-Delbaenrep} for some closed convex set $\Gamma$ of probabilities.\footnote{The representation theorem is due to \citet{ADEH1999}
for finite sample spaces, for general probability spaces see \citet{Delbaen2000} or \citet{FollmerSchied2002}. Its formal statement is not needed for our purposes.} The risk capital required for a portfolio is the worst expected loss over the set $\Gamma$. 

In the context of risk measurement, the model risk measure \eqref{eq-Delbaenrep} is yet another coherent risk measure. %(Coherency follows from the representation theorem of \citet[Thm. 3.2]{Delbaen2000}.) 
Defining the risk measure by the set $\Gamma$ via the representation \eqref{eq-Delbaenrep} is natural when addressing distribution model risk.  Risk measures defined in terms of the profit loss distribution, 
like Value at Risk or Expected Shortfall, rely on a specific distribution model, which may be
misspecified or misestimated. For a fixed portfolio, represented by a pricing function $X$, a different risk factor distribution gives rise to a different profit loss distribution, and therefore to a different risk capital requirement. Expression \eqref{eq-Delbaenrep} measures exactly this model dependence.\footnote{One could object that Expected Shortfall  is coherent and therefore can be represented by eq.~\eqref{eq-Delbaenrep} as a maximum expected loss over some set $\Gamma$ of alternative
distribution models. The set $\Gamma$ equals $\{\Pd={\Pd_0}[.|A]:{\Pd_0}(A)\geq \alpha\}$, which contains distributions so different from $\Pd_0$ that they are hardly plausible to arise from the same historical data by estimation or specification errors. Or, one could represent expected shortfall by eq.~\eqref{eq-Delbaenrep} with $\Gamma$ as in \eqref{eq-Gamma}, taking $D=D_f$ as in \eqref{fdiv} below, with the pathological convex function $f$ equal to $0$ in the interval $[0,1/\alpha]$ and $+\infty$ otherwise \cite[Theorem 4.47]{FollmerSchied2004}. But this $f$ does not meet the assumptions in Section~\ref{sec-relative entropy} and the corresponding $D_f$ is not a divergence in our sense.}

On the other hand, Problem \eqref{eq-Delbaenrep} describes ambiguity averse preferences: A widely used class of preferences allowing for ambiguity aversion are the multiple priors preferences, also known as maxmin expected utility preferences, axiomatised by \citet{GilboaSchmeidler1989}.\footnote{\citet{GilboaSchmeidler1989} worked in the setting of \citet{AnscombeAumann1963} using lottery acts. \citet{Casadesus-Masanell2000} translated their approach to Savage acts. % Multiperiod versions of the multiple priors approach have been developed by \cite{ChenEpstein2002}. 
In the Gilboa-Schmeidler theory the utility of outcomes occurs separately, whereas in our notation the utility is part of the function $X$, which we would interpret as  the utility of outcomes.} % We use the framework of Savage acts to describe the isomporphy between risk measurement and the representation of ambiguity averse preferences.} 
(Another description of ambiguity aversion are the divergence preferences \eqref{DP}.) 
Agents with multiple priors preferences choose acts $X$ with higher worst expected utility, where the worst case is taken over a closed convex set set $\Gamma$ of finitely additive probabilities. The set $\Gamma$ is interpreted as a set of priors held by the agent, and ambiguity is reflected by the multiplicity of the priors. Interpreting the choice of a portfolio as an act, the risk measure representation \eqref{eq-Delbaenrep} and the multiple priors preference representation agree, see \citet{FollmerSchied2002}. A decision maker who ranks portfolios by lower values of some coherent risk measure displays multiple priors preferences. And vice versa, a decision maker with multiple priors preferences acts as if she were minimising some coherent risk measure. 

In the context of the  Gilboa-Schmeidler theory, our results provide explicit expressions for the decision criterion of ambiguity averse decision makers, in the special case that the priors set $\Gamma$ is given by \eqref{eq-Gamma}. Choosing the same $\Gamma$ for all agents may be at odds with a descriptive view of real agents' preferences. But from a normative point of view our choice of $\Gamma$ in \eqref{eq-Gamma} is motivated by general arguments (Section~\ref{sec-relative entropy}). Our results can serve as a starting point for the further analysis of portfolio selection and contingent claim pricing under model uncertainty, extending, among others, work of \cite{Avellaneda1996}, \cite{Friedman2002, Friedman2002b}, \cite{Calafiore2007}.

In the present context, the choice of $\Gamma$ by \eqref{eq-Gamma} with 
$D(\Pd || \Pd_0)$ equal to relative entropy, %=I(\Pd || \Pd_0)$ 
has been proposed by \cite{HansenSargent2001}, see also
\cite{Ahmadi-Javid2011} and \cite{BreuerCsiszar2012}.  \cite{Friedman2002}
also used relative entropy balls as sets of possible models. \cite{HansenSargent2001, HansenSargent2007, HansenSargent2008}, 
% \cite{MatejkaSims2011},
\cite{BarillasHansenSargent2009} and others have used a relative entropy-based
set of alternative models. Their work is set in a multiperiod framework. It
deals with questions of optimal choice, whereas we take the portfolio $X$ as
given. \cite{MaccheroniMarinacciRustichini2006} presented a unified framework
encompassing both the multiple priors preference \eqref{eq-Delbaenrep}
and the divergence preferences \eqref{DP}. They proposed to use
weighted $f$-divergences, which are also covered in our framework. 
%The divergences $D$ they consider
%in \eqref{DP} are weighted $f$-divergences. 
\citet[Theorem
4.2]{BenTalTeboulle2007} showed that their optimised certainty equivalent for
a utility function $u$ can be represented as divergence preference \eqref{DP}
with $D$ equal to the $f$-divergence with the function $f$ satisfying $u(x) =
- f^* (-x)$. For both, % the worst case problem \eqref{eq-Delbaenrep} and the evaluation of the divergence preferences \eqref{DP}, 
the worst case solution is a member
of the same generalised exponential family. This paper makes clear the
reasons. 

Finally but importantly, the work of \cite{Ahmadi-Javid2011} has to be
cited for solutions of \eqref{eq-Delbaenrep} and \eqref{DP}, in case of
relative entropy and of $f$-divergences, in the form of convex 
optimization formulas involving two real variables (one in the case of relative entropy).
The relationship of these results to ours will not be discussed here but
we mention that in \cite{Ahmadi-Javid2011} 
the pathological cases for relative entropy treated in 
Section~\ref{section-RelEnt-pathological} were not addressed, and the 
results for $f$-divergences were obtained under the assumptions that $f$
is cofinite and $X$ is essentially bounded.

\section{Measures of plausibility of alternative risk factor distributions}
\label{sec-relative entropy}
We define divergences between non-negative functions on the state 
space~$\Omega$, which may be any set equipped with a $\sigma$-algebra 
not mentioned in the sequel, and with some measure $\mu$ on that
$\sigma$-algebra.
Here $\mu$ may or may not be a probability measure. Then the
divergence between distributions (probability measures on $\Omega$)
absolutely continuous with respect to $\mu$
is taken to be the divergence between the corresponding density functions.
%\rem{$p( r )=d\Pd/d\mu( r )$ and $p_0( r )=d\Pd_0/d\mu( r )$.}
%If a non-negative function $p$ is a $\mu$-density, denote the corresponding
%distribution by $\Pd$, $p( r )=d\Pd/d\mu( r )$.
In our terminology, a divergence is non-negative and vanishes only
for identical functions or distributions. (Functions which are equal $\mu$-a.e. are
regarded as identical.) A divergence need not be a metric, may be
non-symmetric, and the divergence balls need not form a basis for a
topology in the space of probability distributions.

The relative entropy of two non-negative functions $p,p_0$  is defined as 
\begin{equation*}I(p || \, p_0) :=  \int_\Omega [p ( r ) \log \frac{p( r )}{p_0( r )} - p( r ) + p_0( r )] d\mu( r ).
 \end{equation*}
If $p, p_0$ are $\mu$-densities of probability distributions $\Pd, \Pd_0$ this reduces to the original definition of \citet{KullbackLeibler1951}, 
$$I(\Pd \, || \, \Pd_0)= \int \log \frac{d\Pd}{d\Pd_0}( r ) d\Pd( r
)\quad\mbox{if}\;\;\Pd\ll\Pd_0. $$
If a distribution $\Pd$ is not absolutely continuous with respect to $\Pd_0$,
take $I(\Pd \, || \, \Pd_0) = +\infty$.\footnote{Note that $I(\Pd \, || \, \Pd_0)$
is a less frequent notation for relative entropy than $D(\Pd \, || \, \Pd_0)$,
it has been chosen here because we use the latter to denote any divergence.}

Bregman distances, introduced by \cite{Bregman1967},
and $f$-divergences, introduced by \cite{Csiszar1963, Csiszar1967}, and \citet{AliSilvey1966},
are classes of divergences parametrised by convex functions 
$f:(0,\infty)\rightarrow\R$, extended to $[0,\infty)$ by setting
  $f(0):=\lim_{t\to 0}f(t)$. Below, $f$ is assumed strictly
convex but not necessarily differentiable.

% which are required to be strictly convex, differentiable, non-negative, and to satisfy $f(1)=0$. The derivative $f'$ is assumed to be strictly increasing and to satisfy $f'(1)=0$. %  $f'$ can be interpreted as a utility, since it is strictly increasing. Utilities describe preferences between outcomes, not absolute values of outcomes. The function $f_2(s):=af(s)+bt+c$, for $a>0$ and $b,c\in\R$, yields the utility function $f_2'(s)=af'(s) +b$, which describes the same preferences as $f$. The freedom to choose $a,b,c$ implies that the essential conditions on $f$ are convexity and differentiability. $f(1)=0$ and non-negativity can always be satisfied by an appropriate choice of $b$ and $c$. 

The Bregman distance of non-negative (measurable) functions $p,p_0$ on 
$\Omega$, with respect to a (finite or $\sigma$-finite) measure $\mu$ on 
$\Omega$ is defined by
\begin{equation}
\label{Bfp}
B_{f,\mu}(p,p_0):=\int_\Omega \Delta_f(p( r ), p_0( r )) \mu( dr ),
\end{equation}
where, for $s,t$ in $[0,+\infty)$,
\begin{equation}\label{del}
\Delta_f(s,t):=  \left\{ \begin{array}{ll}
f(s) - f(t) -f'(t)(s-t)& \mbox{if $t>0$ or $t=0$, $f(0) < +\infty$} \\
s\cdot (+ \infty)  & \mbox{if $t=0$ and $f(0)=+\infty$.}
 \end{array} 
 \right.
\end{equation}
If the convex function $f$ is not differentiable at $t$, the right or left derivative is taken for $f'(t)$ according as $s>t$ or $s<t$.

The Bregman distance of distributions $\Pd \ll \mu, \Pd_0 \ll \mu$ is
defined by
\begin{equation}
\label{Bf}
B_{f,\mu}(\Pd,\Pd_0):=B_{f,\mu}\left(\frac{d\Pd}{d\mu}, \frac{d\Pd_0}{d\mu}\right).
\end{equation}
Clearly, $B_{f,\mu}$ is a bona fide divergence whenever $f$ is strictly convex in $(0,+\infty)$.
For $f(s)=s\log s - s + 1$, $B_f$ is the relative entropy $I$. For 
$f(s)=-\log s$, $B_f$ is the Itakura-Saito distance. For $f(s)=s^2$, $B_f$ is 
the squared $L^2$-distance. %If $p$ is the $\Pd_0$-density of a distribution $\Pd$ we get $B_f(\Pd \, || \, \Pd_0) = \int f(d\Pd / d\Pd_0( r )) d\Pd_0( r )$, since $p'( r )=1$, $f(1)=0$ and $f'(1)=0$.

The $f$-divergence between non-negative (measurable) functions $p$ and 
$p_0$
 %, introduced by \citet{Csiszar1963},  
%\citet{AliSilvey1966}, and \citet{Csiszar1967}, 
is defined, when $f$ additionally satisfies $f(s)\ge f(1)=0$,\footnote{This makes sure that \eqref{fdiv} indeed 
defines a divergence between any non-negative functions; if attention is
restricted to probability densities resp. probability distributions, it
suffices to assume that $f(1)=0$.}
by
\begin{equation}\label{fdiv}
D_f(p || p_0):= \int_\Omega f\left(\frac{p(r)}{p_0(r)}\right) p_0(r)\, \mu(dr).
%\:\: \mbox{if} \:\: \Pd\ll\Pd_0.
\end{equation}
At places where $p_0( r )=0$, the integrand by convention is taken to be  
$p( r )\lim_{s\to\infty} f(s)/s$.
The $f$-divergence of distributions $\Pd \ll \mu, \Pd_0 \ll \mu$, defined
as the $f$-divergence of the corresponding densities, does not depend on
$\mu$ and is equal to
 
\begin{equation}\label{fdi}
D_f(\Pd || \Pd_0):= \int_\Omega f\left(\frac{d\Pd_a}{d\Pd_0}\right) d\Pd_0 
%\:\: \mbox{if} \:\: \Pd\ll\Pd_0.
+\Pd_s(\Omega)\lim_{s\to\infty}\frac{f(s)}{s},
\end{equation}
where $\Pd_a$ and $\Pd_s$ are the absolutely continuous and singular 
components of $\Pd$ with respect to $\Pd_0$. Note that if $f$ is
cofinite, i.e., if the limit in \eqref{fdi} is $+\infty$,
then $\Pd\ll\Pd_0$
is a necessary condition for the finiteness of $D_f(\Pd || \Pd_0)$,
while otherwise not. 
%Note also that if $\Pd\ll\Pd_0$ and $f$ is differentiable at $t=1$ then 
%$$ D_f(\Pd || \Pd_0)= \int f(d\Pd / d\Pd_0) 
%d\Pd_0,=B_{f, \Pd_0}(\Pd \, || \, \Pd_0)    $$ since $f(1)=0$ and $f'(1)=0$.

For $f(s)=s\log s - s + 1$, $D_f$ is the relative entropy. For $f(s)=-\log s + s + 1$, $D_f$ is the reversed relative entropy. For $f(s)=(\sqrt{s}-1)^2$, $D_f$ is the squared Hellinger distance. For $f(s)=(s-1)^2/2$, $D_f$  is the relative Gini concentration index. 
For more details about $f$-divergences see \citet{LieseVajda1987}. 

Relative entropy appears the most versatile divergence measure for probability distributions or non-negative functions, extensively used in diverse fields including statistics, information theory, statistical physics, see e.g. \citet{Kullback-book}, \citet{CsiszarKoerner2011}, \citet{Jaynes1957a}. For its applications in econometrics, see \citet{GolanJudgeMiller1996} or \citet{Grechuk2009}. In the context of this paper, \cite{HansenSargent2001} have used expected value minimization over relative entropy balls. Arguments for \eqref{eq-Gamma} with any $f$-divergence in the role of $D$, or more generally with a weighted $f$-divergence involving a (positive) weight function $w( r )$ in the integral in \eqref{fdiv}, have been put forward by \citet{MaccheroniMarinacciRustichini2006}. Results of \cite{Ahmadi-Javid2011} indicate advantages of relative entropy over other $f$-divergences also in this context.
% Also \cite{Ahmadi-Javid2011} studied risk measures \eqref{eq-Delbaenrep}, with \eqref{eq-Gamma} defined by an $f$-divergence, pointing out adventages of using relative entropy. 
In another context, \citet{GrunwaldDawid2004} argue that distances between distributions might be chosen in a utility dependent way. Relative entropy is natural only for decision makers with logarithmic utility. Picking up this idea, for decision makers with non-logarithmic utility one might define the radius in terms of some utility dependent distance. We are unaware of references employing \eqref{eq-Gamma} with Bregman distances, although this would appear natural, particularly as Bregman distances have a beautiful interpretation as measuring the expected utility losses due to the convexity of $f$.

In the context of inference, the method of maximum entropy (or relative
entropy minimization)  
is distinguished by axiomatic considerations. \citet{ShoreJohnson1980},
\citet{ParisVencovska1990}, % \citet{JonesByrne1990}, 
and \citet{Csiszar1991} showed that it is the only method that satisfies certain 
intuitively desirable postulates.
  %\paragraph{Arguments for other divergence balls}
Still, %as \citet{Uffink1995,Uffink1996} argued, 
relative entropy cannot be singled out as providing the only
reasonable method of inference. \citet{Csiszar1991} determined what 
alternatives (specifically, Bregman distances and $f$-divergences) come
into account if some postulates are relaxed. In the context of 
measuring risk or evaluating preferences under ambiguity aversion, axiomatic results distinguishing relative
entropy or some other divergence are not available.

%Arguments for \eqref{eq-Gamma} with any $f$-divergence in the role of $D$, or more generally with a weighted $f$-divergence involving a (positive) weight function $w( r )$ in the integral in \eqref{fdi}, have been put forward by \citet{MaccheroniMarinacciRustichini2006}. \citet{GrunwaldDawid2004} argue that distances between distributions might be chosen in a utility dependent way. Relative entropy is natural only for decision makers with logarithmic utility. Picking up this idea, for decision makers with non-logarithmic utility one might define the radius in terms of some utility dependent distance. 
 
An objection against the choice of the set $\Gamma$ in \eqref{eq-Gamma}
with $D$ equal to relative entropy or a related divergence should also be
mentioned. It is that all distributions in this set are absolutely continuous with respect to $\Pd_0$. 
In the literature of the subject, even if not working with divergences, it is a
rather common assumption that the set of feasible distributions is
dominated; one notable exception is \cite{Cont2006}.  Sometimes the assumption that $\Gamma$ is dominated is hard to justify. 
For example, in a multiperiod setting 
where $\Omega$ is the canonical space of continuous paths and $\Gamma$ is a set of martingale laws for the canonical process, corresponding to different scenarios of volatilities, this $\Gamma$ is typically not dominated (see \citet{NutzSoner2012}). 
Or, if we use a continuous default distribution $\Pd_0$, can we 
always be sure that the data generating process is not discrete? And should 
it not be possible to approximate in some appropriate sense a continuous 
distribution by discrete ones? %$f$-divergences do not share this problem if $f$ is not cofinite, i.e. if $f'(t)$ does not go to infinity as $t$ goes to infinity. $f$-divergences are defined also when $\Pd \ll \Pd_0$ in \eqref{fdi} does not hold. The definition for this case is 
%\begin{equation}\label{fdiv}
%D_f(\Pd || \Pd_0) := \int_\Omega f\left(\frac{d\Pd_a}{d\Pd_0}\right) d\Pd_0 + \%Pd_s(\Omega) \lim_{t\rightarrow +\infty} \frac{f(t)}{t},
%\end{equation}
%where $\Pd=\Pd_a+\Pd_s$ is the decompostion of $\Pd$ into absolutely
%continuous and singular parts with respect to $\Pd_0$. If $f$ is cofinite,
%i.e. if $\lim_{s\rightarrow +\infty} f(s)/s = +\infty$, then $D_f(\Pd ||
%\Pd_0) = +\infty$ if $\Pd$ is not absolutely continuous with respect to
%$\Pd_0$, hence $\Gamma$ above coincides with the $f$-divergence ball. If $f$
%is not cofinite then the $f$-divergence ball contains also distributions not
%absolutely continuous with respect to $\Pd_0$.%, thus not in $\Gamma$. 

If an $f$-divergence with a non-cofinite $f$ is used, then 
%Under these circumstances 
the set $\Gamma$ of alternative distributions is not dominated, see 
\eqref{fdi}. But since all distributions singular
%in $\Gamma$ not absolutely continuous with respect 
to $\Pd_0$ have the same $f$-divergence from $\Pd_0$, % In this sense 
even $f$-divergences with non-cofinite $f$ are not appropriate to describe the approximation of a continuous distribution by discrete distributions.
Bregman distances have a similar shortcoming. In practice, this
objection does not appear a serious obstacle, for
the set $\Gamma$ of theoretical alternatives may be extended by distributions 
close to them 
in an appropriate sense involving closeness of expectations, which
negligibly changes the theoretical risk value \eqref{eq-Delbaenrep}. 

%We are unaware of references employing \eqref{eq-Gamma} with Bregman distances, although this would appear natural. 
%Bregman distances have a beautiful interpretation as measuring the expected utility losses due to the convexity of $f$. But 
%they are also inappropriate to describe the approximation of a continuous distribution by discrete distributions: Any distribution supported by a finite number of points is is outside the ``sufficiently small'' Bregman balls centered at a continuous distribution. (In Section~\ref{sec-CIF} the threshold value $\Flim$ makes precise the meaning of ``sufficiently small''.) 

\section{Intuitive relation of worst case risk and  maximum entropy inference}\label{sec-intuition}
The purpose of this section is to develop intuition on the relation between Problem~\eqref{eq-Delbaenrep} and the maximum entropy problem. Let us consider the mathematically simplest case of
%In this section we discuss the solution to 
Problem \eqref{eq-Delbaenrep}, when $\Gamma$ is a sufficiently small relative 
entropy ball. Then Problem~\eqref{eq-Delbaenrep} requires the evaluation of
\begin{equation}
\label{eq-genMaxLoss}
\inf_{\Pd: I(\Pd \, || \, \Pd_0)\leq k} E_\Pd (X) =: V(k),
\end{equation}
for sufficiently small $k$. 
We follow \cite{BreuerCsiszar2012}, using techniques familiar in the theory of exponential families, see
\citet{Barndorff-Nielsen1978}, and large deviations theory, see \citet{DemboZeitouni}. 
The meaning of `sufficiently small' will be made
precise later in this section. The cases when 
%The results reported are from \cite{BreuerCsiszar2012}. In
%Section~\ref{section-RelEnt-pathological} we will provide the solution in 
the relative entropy ball is not `sufficiently small' will be treated in
Section~\ref{section-RelEnt-pathological}.
%Additionally, this section introduces in a well-known context concepts like the $\Lambda$-function and the exponential family, which play an important role in later sections.

Observe that Problem~\eqref{eq-Delbaenrep} with $\Gamma$ a relative entropy 
ball is ``inverse'' 
to a problem of maximum              entropy inference. If an unknown
distribution $\Pd$ had to be inferred when the available information specified
only a feasible set of distributions, and a distribution ${\Pd_0}$ were given
as a prior guess of $\Pd$, the maximum entropy\footnote{This name refers to
  the special case when %the prior guess 
  ${\Pd_0}$ is the uniform distribution;
  then minimising $I(\Pd\, || \, {\Pd_0})$ is equivalent to maximising the
  Shannon differential entropy of $\Pd$.} principle would suggest to 
infer the feasible distribution $\Pd$ which minimizes  
$I(\Pd\, || \, {\Pd_0})$. In particular, if the feasible distributions were 
those with $E_\Pd(X)=b$, for a constant $b$, we would arrive at the problem
\begin{equation}
\label{eq-MaxEnt}
\inf_{\Pd: E_\Pd (X)=b}I(\Pd\, || \, {\Pd_0}).  
\end{equation}
Note that the objective function of problem \eqref{eq-Delbaenrep} is the
constraint in the maximum entropy problem~\eqref{eq-MaxEnt}, and vice versa
(Fig.~\ref {fig-duality}). It is therefore intuitively expected that (taking
$k$ and $b$ suitably related) both problems are solved by the same
distribution $\Pbar$,
\begin{equation}
\label{eq-conjecture}
\arg \min_{\Pd: I(\Pd\, || \, {\Pd_0})\leq k} E_\Pd (X) = \arg \min_{\Pd: E_\Pd (X)=b}I(\Pd \, || \,  {\Pd_0})=:\Pbar,
\end{equation}
see Fig.~\ref{fig-duality}. The literature on the maximum entropy problem
establishes that (under some regularity conditions) the solution $\Pbar$ is a
member of the exponential family of distributions $\Pd(\theta)$ with 
canonical statistic $X$, which have a ${\Pd_0}$-density 
\begin{equation}
\label{eq-wcs}
\frac{d\Pd(\theta)}{d{\Pd_0}}( r ):=\frac{e^{\theta X( r )}}{\int e^{\theta X( r )} d{\Pd_0}( r )}=e^{\theta X( r )-\Lambda(\theta)},  
\end{equation}
where $\theta\in\R$ is a %negative 
parameter and the function $\Lambda$ is defined as 
\begin{equation}
\label{eq-def-Lambda}
\Lambda(\theta):=\log\int e^{\theta X( r )} d{\Pd_0}( r ).
\end{equation}
%If the profit function $X$ is clear from the context, we will simply write
%$\Lambda(\theta)$. 
Among actuaries the distributions from the exponential family are often 
referred to as Esscher transforms.

\begin{figure}
   \centering
    \includegraphics[width=7.6cm]{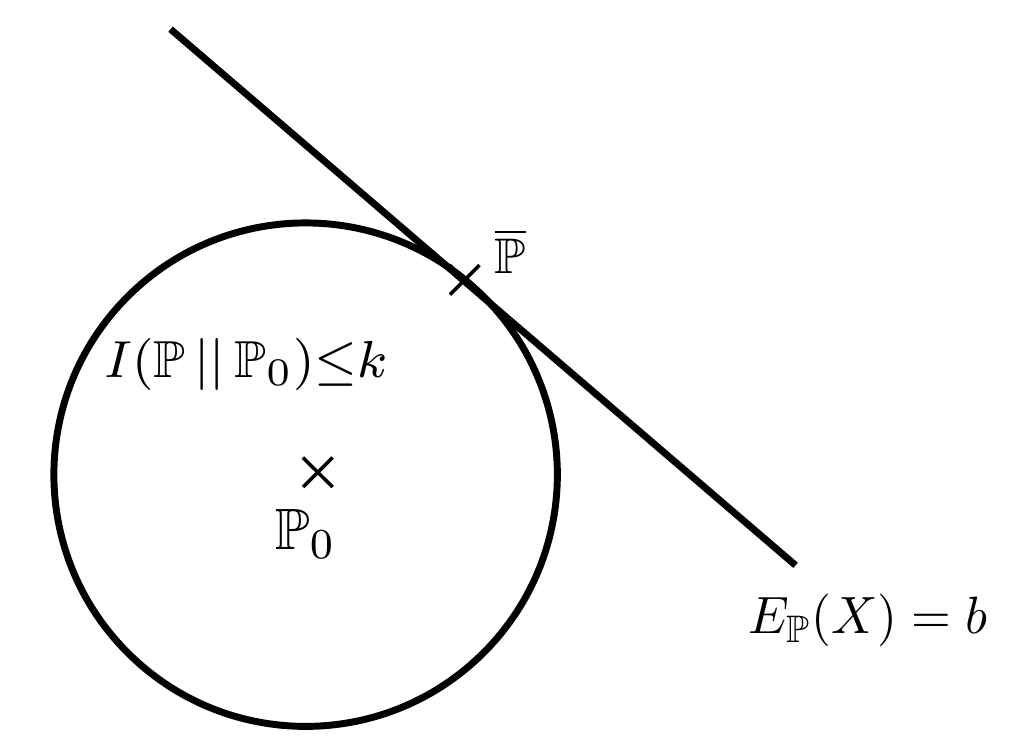}
%    \hspace{1cm}
%    \includegraphics[width=5.4cm,height=4.9cm]{Lambda-dual2.pdf}
\caption{\label{fig-duality}
{\bf Relation of Worst Case and Maximum Entropy problem.} What is the objective 
function in %the worst case 
problem~\eqref{eq-Delbaenrep} is the constraint in the maximum %relative 
entropy problem~\eqref{eq-MaxEnt}, and vice versa.}
\end{figure}

For members $\Pd(\theta)$ of the exponential family, the expected profit can be written as
\begin{equation}
\label{eq-Lambda'}
E_{\Pd(\theta)} (X)=\int X( r )\exp(\theta X( r ) - \Lambda(\theta))d{\Pd_0}( r ) = \Lambda'(\theta),
\end{equation}
%(where $\Lambda'(\theta)$ is the partial derivative of $\Lambda(\theta,X)$
%with respect to $\theta$)
and the relative entropy to $\Pd_0$ is 
\begin{eqnarray}
I(\Pd(\theta)\, || \, {\Pd_0}) & = & \int \log \frac{d\Pd(\theta)}{d{\Pd_0}}( r )d\Pd(\theta)( r ) = \int (\theta X( r ) - \Lambda(\theta)) d\Pd(\theta)( r ) \nonumber \\
& = & \theta E_{\Pd(\theta)} (X) -\Lambda(\theta) =  \theta\Lambda'(\theta)-\Lambda(\theta). \label{eq-A2}
\end{eqnarray}

{\em If} the identity \eqref{eq-conjecture} holds and the solution of Problem~\eqref{eq-genMaxLoss} is from the exponential family, then one can determine which member of the exponential family solves the problem, by solving the equation 
\begin{equation}
\label{eq-def-theta}
\theta \Lambda'(\theta) - \Lambda(\theta)=k
\end{equation}
for $\theta$. Typically, \eqref{eq-def-theta} has both a positive and a
  negative solution, and the corresponding $\Pd(\theta)$ is the maximiser
  resp. minimiser of  $E_{\Pd} (X)$ subject to $I(\Pd\, || \, {\Pd_0})\le
  k$. Call the negative solution $\thbar$. The solution to Problem~\eqref{eq-genMaxLoss} can then be expressed in terms of the $\Lambda$-function:
$$\inf_{\Pd: I(\Pd\, || \, {\Pd_0})\leq k} E_\Pd (X) = \inf_{\theta: \theta\Lambda'(\theta)-\Lambda(\theta)\leq k} \Lambda'(\theta) = \Lambda'(\thbar),$$
(The last equality follows from the convexity of $\Lambda$.)
This solution is illustrated in Fig.~\ref{fig-theorems 1 and 2}. The worst expected profit $V(k)$ is the slope of the tangent to the curve $\Lambda(\theta)$ passing through $(0,-k)$. $\thbar$ is the $\theta$-coordinate of the tangent point. From the figure it is obvious that $\thbar\Lambda'(\thbar)-\Lambda(\thbar)=k$.

\begin{figure}
   \centering
    \includegraphics[width=7.4cm,height=6cm]{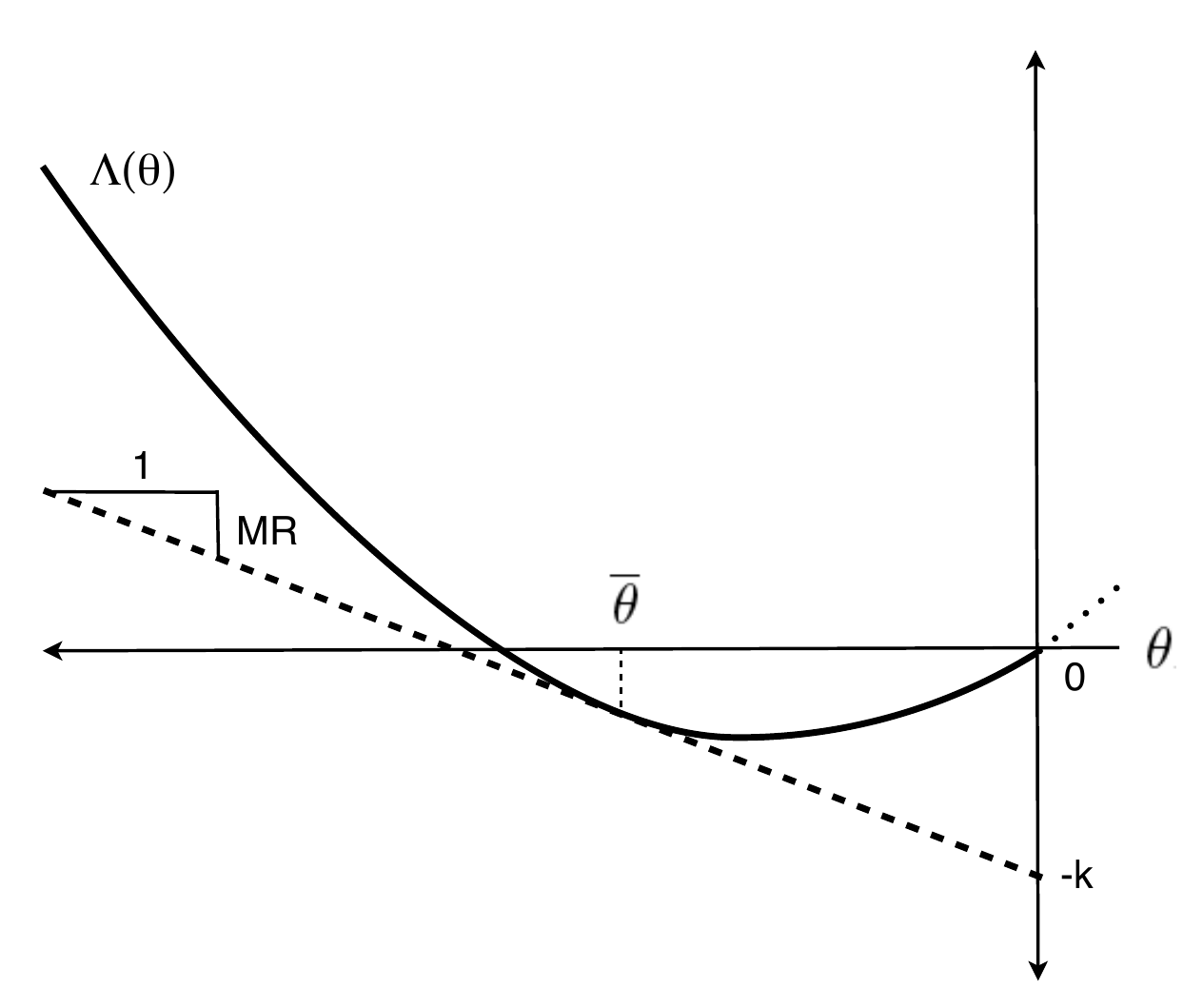}
%    \hspace{1cm}
%    \includegraphics[width=5.4cm,height=4.9cm]{Lambda-dual2.pdf}
\caption{\label{fig-theorems 1 and 2}
{\bf Solution of the worst case problem problem from the $\Lambda$-function.} The optimal value achieved for Problem~\eqref{eq-genMaxLoss} is the slope of the tangent to the curve $\Lambda(\theta)$ passing through $(0,-k)$. $\thbar$ is the $\theta$-coordinate of the tangent point.}
\end{figure}

So far the intuition about the solution in what one could call the generic
case. It requires two important assumptions: Identity \eqref{eq-conjecture}
should hold %(so that the solution is from the exponential family) 
and the equation \eqref{eq-def-theta}
%$\thbar\Lambda'(\thbar)-\Lambda(\thbar)=k$ 
should have a (unique) negative solution 
$\thbar$. \citet{BreuerCsiszar2012} give precise conditions under which the 
%two assumptions hold and the 
solution is indeed of the generic form above. The first condition is relevant when $X$ is essentially bounded
below, the other two when it is not: 
\begin{itemize}
\item[(i)] If $\essinf(X)$ is finite, %assume $k$ is smaller than 
$k<k_{\max}:= - \log {\Pd_0}(\{r: X( r )=\essinf(X)\}).$
\item[(ii)] %Assume 
$\thmin:=\inf \{\theta: \Lambda(\theta)<+\infty\} <  0$,
\item[(iii)] If $\thmin,\; \Lambda(\thmin)$, and $\Lambda'(\thmin)$ are all 
finite %, assume $k$ does not exceed 
then $k\le\thmin \Lambda'(\thmin) - \Lambda(\thmin).$
\end{itemize}

The concepts used above are in close analogy to statistical mechanics. 
%First, as pointed out above, the optimisation problem of finding the worst
%case model is `dual' to the method of maximum entropy. Second, almost all
%quantities in the model risk analysis have counterparts in statistical
%mechanics:
The risk factor vector $r$ is the counterpart of the phase space points.
The pricing function $X$ is the counterpart of the energy function.
$\Lambda$ is the counterpart of the logarithm of the partition function $Z$.
 $\theta$ is the counterpart of the inverse temperature parameter
 $\beta=1/kT$. The worst case distribution~\eqref{eq-wcs} is the counterpart 
 of the canonical distribution.
%The worst case distribution~\eqref{eq-wcs} is the counterpart of the
%canonical distribution. $\theta$ is the counterpart of the inverse
%temperature parameter $\beta=1/kT$. The profit function $X$ is the
%counterpart of the energy function. The risk factor vector $r$ is the
%counterpart of the phase space points. $\Lambda$ is the counterpart of the
%logarithm of the partition function $Z$.

\section{Maximum Loss over relative entropy balls: The pathological cases}
\label{section-RelEnt-pathological}
Now let us turn to the solution of Problem~\eqref{eq-genMaxLoss} in the pathological case where $\Gamma$ is a large relative entropy ball, so that one of the conditions (i)-(iii) is violated.

First consider the case that assumption (i) above is violated, where the loss
is essentially bounded and the sphere is not ``sufficiently small''. Long bond
portfolios are examples for this case. In this case
equation~\eqref{eq-def-theta} has no negative solution. The shape of the $\Lambda$-function is displayed in Fig.~\ref{fig-pathological1}.

\begin{proposition}
\label{prop-pathological1}
If $\essinf(X)$ is finite, and $k \ge  \kmax$ (defined in (i), Section~\ref{sec-intuition})
% - \log({\Pd_0}(\{r: X( r )=\essinf(X)\})),$ 
then the solution to Problem~\eqref{eq-genMaxLoss} is $V(k)=\essinf(X)$. The worst case distribution $\Pbar$ has the $\Pd_0$-density
\begin{equation*}
\frac{d\Pbar}{d{\Pd_0}}( r ) := \left\{ \begin{array}{ll}
 1/\beta & \mbox{if $X( r )=\essinf(X)$} \\
 0  & \mbox{otherwise,}
 \end{array} 
 \right.
\end{equation*}
where $\beta= {\Pd_0}(\{r: X( r )=\essinf(X)\})$.
\end{proposition}

\begin{proof}
The distribution $\Pbar$ satisfies 
$$I(\Pbar \, || \,  {\Pd_0})=\int\log \frac{d\Pbar}{d{\Pd_0}} d\Pbar = -\log \beta,$$
hence $I(\Pbar \, || \,  {\Pd_0})\leq k$ if $k \geq - \log \beta.$
Then $V(k)\leq E_{\Pbar}(X)$. Trivially $V(k)\geq \essinf(X)$. The claim $V(k)=\essinf(X)$ follows.
\end{proof}

\begin{figure}[htbp]
\begin{center}
   \centering
    \includegraphics[width=5.9cm,height=5.5cm]{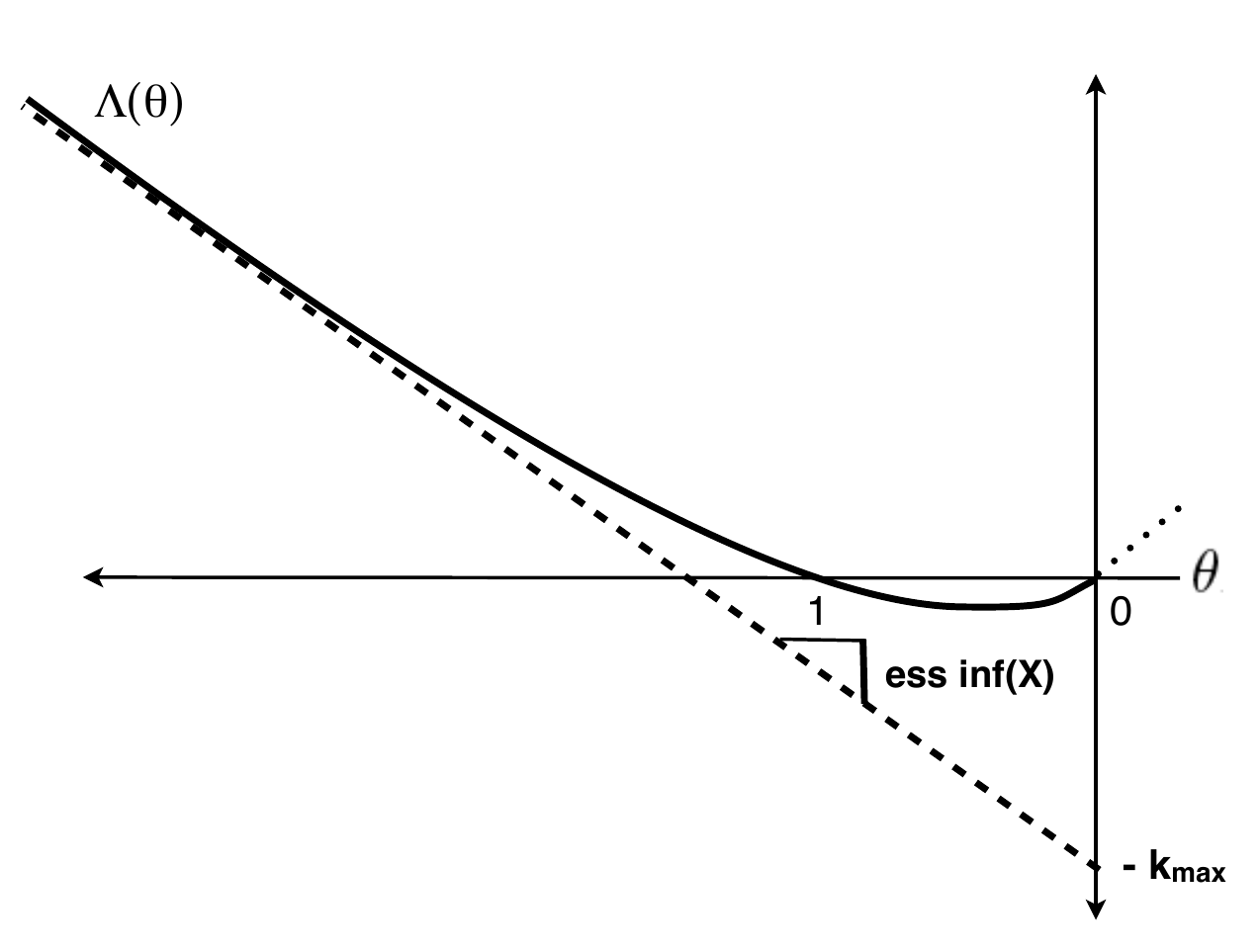}
\caption{\label{fig-pathological1}
The pathological case of Proposition~\ref{prop-pathological1}.}
\end{center}
\end{figure}

Next consider the pathogolical case that assumption (ii) above is violated so that $\thmin=0$, and thus $ \Lambda(\theta) = +\infty$ for all $\theta < 0$. 
\begin{proposition}
\label{prop-pathological2}
If $\thmin$ defined in (ii) of Section~\ref{sec-intuition} equals zero, then the solution to Problem~\eqref{eq-genMaxLoss} is $V(k)= - \infty$ for all $k>0$. 
% Neither eq.~\eqref{eq-Lambda*-k} nor eq.~\eqref{eq-def-theta} has a solution.
\end{proposition}
\begin{proof}
Let $\beta_{m,n}:={\Pd_0}(\{r:-n\leq X( r )\leq m\})$ and consider the measures ${\Pd}_{m,n} \ll {\Pd_0}$ with
\begin{equation*}
\frac{d{\Pd}_{m,n}}{d{\Pd_0}}( r ) := \left\{ \begin{array}{ll}
 1/\beta_{m,n} & \mbox{if $-n\leq X( r )\leq m$} \\
 0  & \mbox{otherwise.}
 \end{array} 
 \right.
\end{equation*}
Obviously $I(\Pd_{m,n} || \Pd_0) = - \log\beta_{m,n}.$
For any $\Pd \ll {\Pd}_{m,n}$,
\begin{equation*}
I(\Pd\, || \, {\Pd_0})=\int \log(\frac{d\Pd}{d{\Pd}_{m,n}}\frac{d{\Pd}_{m,n}}{d{\Pd_0}})d\Pd = I(\Pd\, || \, {\Pd}_{m,n})-\log \beta_{m,n}
\end{equation*}
is arbitrarily close to $I(\Pd\, || \, {\Pd}_{m,n})$ if $m$ and $n$ are sufficiently large. Hence to prove that
$V(k)= - \infty$ for all $k>0$, it suffices to find to any given $m$ and sufficiently large $n$ distributions $\Pd \ll {\Pd}_{m,n}$ with $I(\Pd\, || \, {\Pd}_{m,n})$ arbitrarily close to zero and $E_\Pd(X)$ arbitrarily low.

In the rest of this proof, $m$ is fixed and $n$ will go to $+\infty$. Define $\Pd$ and $\Lambda_{m,n}$ by 
$$\frac{d\Pd}{d{\Pd}_{m,n}}( r ):=\frac{e^{\theta X( r )}}{\int e^{\theta X( r )} d{\Pd}_{m,n}( r )}=:e^{\theta X( r )-\Lambda_{m,n}(\theta)}$$
for any $\theta < 0$.
$\Pd$ and $\Lambda_{m,n}$ depend on $\theta$. 
As in \eqref{eq-Lambda'}, $E_\Pd(X)=\Lambda'_{m,n}(\theta)$ and $I(\Pd\|{\Pd}_{m,n})=\theta\Lambda'_{m,n}(\theta)-\Lambda_{m,n}(\theta)$ for any $\theta < 0$.
For each $\theta$, 
\begin{equation}\label{eq-Lambda_mn}
-\theta m\geq \Lambda_{m,n}(\theta)=\int_\theta^0 \Lambda'_{m,n}(\xi)d\xi\geq - \theta\Lambda'_{m,n}(\theta), 
\end{equation}
since $\Lambda'_{m,n}$ is increasing.
For fixed $\theta<0$, $\Lambda_{m,n}(\theta) \rightarrow \infty$ as $n\rightarrow \infty$ since $\Lambda(\theta)= \infty$ by assumption.
By \eqref{eq-Lambda_mn} it follows that $\Lambda'_{m,n}(\theta)\rightarrow - \infty$ as $n\rightarrow \infty$, and hence there exists a sequence $\theta_n\uparrow 0$ such that $\Lambda'_{m,n}(\theta_n)\rightarrow - \infty$ and $\theta_n \Lambda'_{m,n}(\theta_n)\rightarrow 0$ as $n\rightarrow\infty$. By inequality~\eqref{eq-Lambda_mn}, this implies $|\Lambda_{m,n}(\theta_n)|\rightarrow 0$ and hence $I(\Pd\|{\Pd}_{m,n})\rightarrow 0$ as $n\rightarrow \infty$.
This completes the proof that, for $\Pd$ defined with $\theta=\theta_n$, $E_\Pd (X)$ will be arbitrarily low and  $I(\Pd\, || \, {\Pd}_{m,n})$ arbitrarily small but positive.
\end{proof}

Finally consider the case that both $\Lambda(\thmin)$ and 
$\Lambda '(\thmin)$ are finite, but the sphere is not ``sufficiently small''. The shape of the $\Lambda$-function is displayed in Fig.~\ref{fig-pathological3}.

\begin{proposition}
\label{prop-pathological3}
If $-\infty < \thmin < 0$, ($\thmin$ is defined in (ii) of Section~\ref{sec-intuition}), and both $\Lambda(\thmin)$ and 
$\Lambda '(\thmin)$ are finite, and additionally 
$k>\thmin \Lambda'(\thmin) - \Lambda(\thmin)$,
%$k>k_{\max}$, 
then %the solution to Problem~\eqref{eq-genMaxLoss} equals 
\begin{equation}\label{eq-MK-thmax}
V(k) = (k + \Lambda(\thmin))/\thmin,
\end{equation}
but there is no distribution achieving this value; the infimum in Problem~\eqref{eq-genMaxLoss} is not a minimum. 
\end{proposition}

\begin{figure}[htbp]
\begin{center}
   \centering
\includegraphics[width=10.9cm,height=7.5cm]{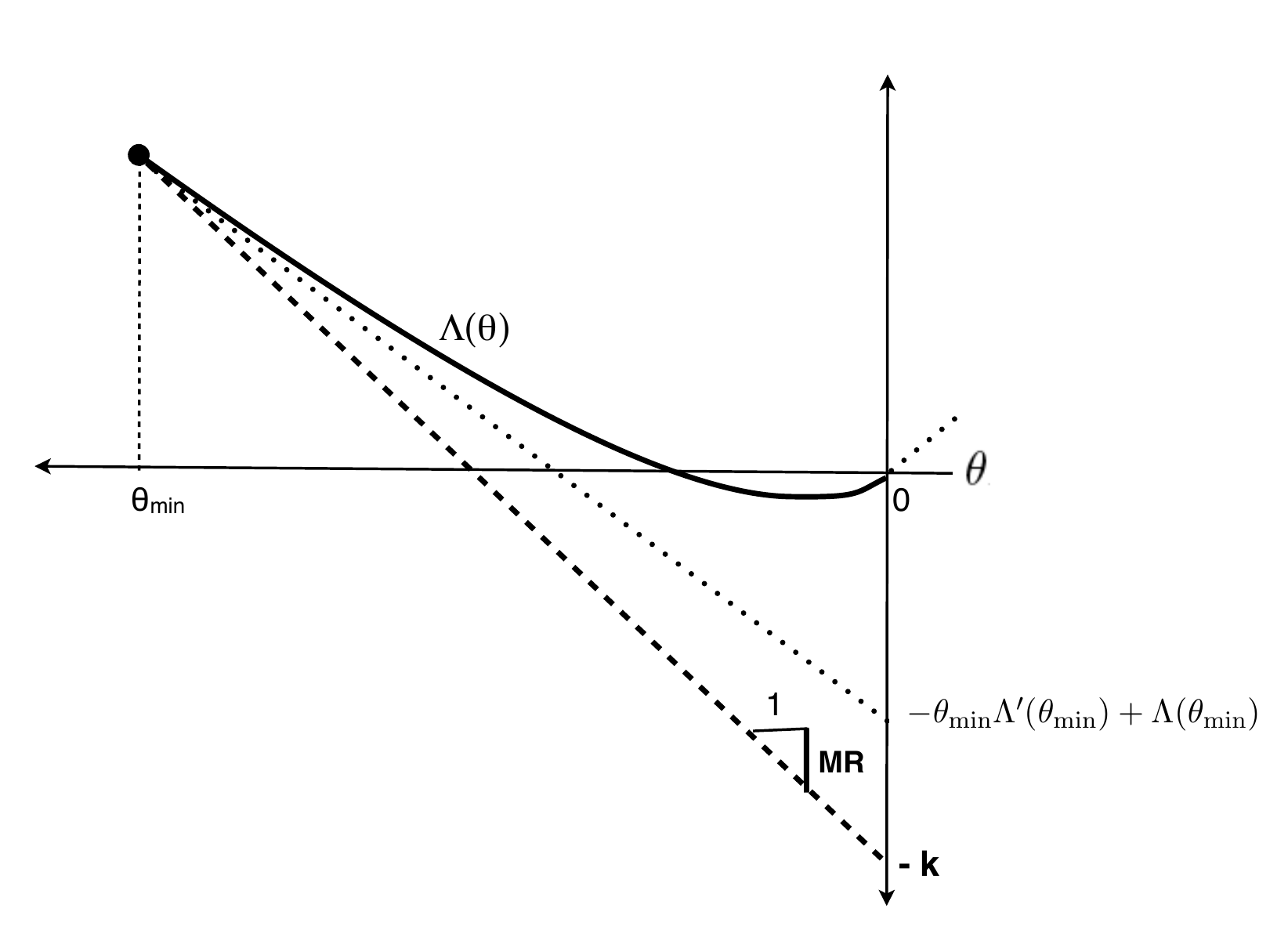}
\caption{\label{fig-pathological3}
The pathological case of Proposition \ref{prop-pathological3}.}
\end{center}
\end{figure}

\begin{proof}
%The \rem{convex conjugate of $\Lambda$ defined by}
%\begin{equation}
%\label{eq-def-Lambda*}
%\Lambda^*(x):=\sup_{\theta}(\theta x - \Lambda(\theta)),
%\end{equation}
%called the convex conjugate of $\Lambda(\theta)$, 
%is also convex and lower semicontinuous on $\R$. Clearly, 
%\begin{equation}
%\label{eq-Lambda*}
%\Lambda^*(x) = \theta x -  \Lambda(\theta) \:\: \mbox{if} \:\: x=\Lambda'(\theta).
%\end{equation}
%However, for some $x$ perhaps no $\theta$ satisfies $x=\Lambda '(\theta)$.
%Still, there exists a unique $x$ satisfying
% \begin{equation}
%\label{eq-Lambda*-k}
%\Lambda^*(x)=k \:\: \mbox{and} \:\: x < E_{\Pd_0}(X),
%\end{equation}
%but this $x$ does not equal $\Lambda '(\theta)$ for any $\theta$.

Define $\Pd(\thmin)$ as in \eqref{eq-wcs} with $\thmin$ in the place of  
$\theta$. Then %for all $\Pd\ll \Pd(\thmin)$ we have
\begin{eqnarray}
I(\Pd\| {\Pd_0}) & = & \int \log\left(\frac{d\Pd}{d\Pd(\thmin)}\frac{d\Pd(\thmin)}{d{\Pd_0}} \right) d\Pd \nonumber \\
& = & I(\Pd\|\Pd(\thmin)) + \int\log( \exp( \thmin X( r ) - \Lambda(\thmin))d\Pd( r ) \nonumber \\
& = & I(\Pd\|\Pd(\thmin)) + \thmin E_\Pd(X) -  \Lambda(\thmin) \label{eq-id}
\end{eqnarray}
for all $\Pd\ll \Pd(\thmin)$. Hence, if $I(\Pd \, || \,  {\Pd_0})\leq k$ 
then (using $\thmin<0$)
\begin{equation}
\label{eq-ub}
E_\Pd(X) \geq (k + \Lambda(\thmin) - I(\Pd\|\Pd(\thmin)))/\thmin,
\end{equation}
proving that $V(k)\geq (k + \Lambda(\thmin))/\thmin.$
To show that equality holds, apply the result of Proposition~\ref{prop-pathological2} to $\Pd(\thmin)$ in the role of ${\Pd_0}$, then the role of $\Lambda(\theta)$ is played by
$$\Labar(\theta):=\log\int e^{\theta X( r )}d\Pd(\thmin)( r ) = \Lambda(\theta
+ \thmin) - \Lambda(\thmin).$$ Clearly, $\Labar(\theta)=\infty$ for all
$\theta < 0$, hence by Proposition~\ref{prop-pathological2} there exist
distributions $\Pd'$ with $I(\Pd'\|\Pd(\thmin))$ arbitrarily small but
positive and $E_{\Pd'}(X)$ arbitrarily low. Then, for any small $\epsilon>0$,
a suitable linear combination $\Pd$ of $\Pd'$ and $\Pd(\thmin)$ satisfies
$E_\Pd(X)=(k+\Lambda(\thmin)-\epsilon)/\thmin$ and
$I(\Pd\|\Pd(\thmin))<\epsilon$. For this $\Pd$, eq. \eqref{eq-id} implies that
$I(\Pd\, || \, {\Pd_0})\leq k$ and the claim $V(k)\leq (k +
\Lambda(\thmin))/\thmin$ follows. This proves that $V(k) = (k + \Lambda(\thmin))/\thmin$.

\eqref{eq-ub} implies that in Problem~\eqref{eq-genMaxLoss} the supremum is not attained because $I(\Pd \, || \, \Pd(\thmin))$ is strictly positive when $E_\Pd(X) < \Lambda'(\thmin) = E_{\Pd(\thmin)}(X)$.  \end{proof}

\begin{remark} \label{r-Lambda*} 
Consider the convex conjugate of $\Lambda$ defined by
\begin{equation}
\label{eq-def-Lambda*}
\Lambda^*(x):=\sup_{\theta}(\theta x - \Lambda(\theta)),
\end{equation}
%called the convex conjugate of $\Lambda(\theta)$, 
which is a convex, lower semicontinuous function on $\R$. Clearly, 
\begin{equation}
\label{eq-Lambda*}
\Lambda^*(x) = \theta x -  \Lambda(\theta) \:\:\:\: \mbox{if}\:\: \:\: x=\Lambda'(\theta).
\end{equation}
However, for some $x$ perhaps no $\theta$ satisfies $x=\Lambda '(\theta)$.
%Still, there exists a unique $x$ satisfying
% \begin{equation}
%\label{eq-Lambda*-k}
%\Lambda^*(x)=k \:\: \mbox{and} \:\: x < E_{\Pd_0}(X),
%\end{equation}
%but this $x$ does not equal $\Lambda '(\theta)$ for any $\theta$.
In the generic case, when the assumptions (i)-(iii) of
Section~\ref{sec-intuition} are met, the optimal value attained in
Problem~\eqref{eq-Delbaenrep} is equal to $x=\Lambda '(\thbar)$ for $\thbar$
satisfying \eqref{eq-def-theta}, which $x$ is the unique solution of 
\begin{equation}\label{eq-Lambda*-k}
\Lambda^*(x)=k \:\: \mbox{and} \:\: x < E_{\Pd_0}(X).
\end{equation} 
The proof of Proposition~\ref{prop-pathological3} establishes that $V(k)$ always equals the solution of \eqref{eq-Lambda*-k} when it exists, even if \eqref{eq-def-theta} does not have a solution.
\end{remark}
                                    
% In the pathological case of Proposition~\ref{prop-pathological3} the function $\deltabar(\epsilon)$ equals
% $$\deltabar(\epsilon)=\thmax\epsilon$$ since
% \begin{equation*}
% k  \geq  I(\Pd || \Pd_0)  = I(\Pd || \Pd(\thmax)) + \thmax E_\Pd(-X) -  \Lambda(\thmax) 
% \end{equation*}
% so that 
% \begin{eqnarray*}
% I(\Pd || \Pd(\thmax)) & \leq & k+ \Lambda(\thmax) - \thmax E_\Pd(-X) \\
% & = & \thmax (\MaxLoss - E_\Pd(-X)) \\
% & \leq & \thmax \epsilon,
% \end{eqnarray*}
% where in the second line we used \eqref{eq-MK-thmax}.

\section{A more general framework}
\label{sec-CIF}
Now we construct a unified framework that covers the choices of $\Gamma$ in
\eqref{eq-Gamma} when $D$ is an $f$-divergence or a Bregman distance, as well
as others. In this framework, $\Gamma$ is chosen as a set of 
probability measures $P\ll\mu$ (where $\mu$ is a given measure on $\Omega$, 
finite or $\sigma$-finite) of the form %whose density $p=d\Pd/d\mu$ satisfies
\begin{equation}\label{hbek}
\Gamma = \{\Pd\ll \mu:\; p=d\Pd/d\mu\;\; \text{satisfies}\;\;H( p )\leq k \},
\end{equation}
where $H$ is a convex integral functional defined as 
\begin{equation}
\label{H}
H( p ) :=\int_\Omega \beta(r, p( r ))\mu(dr),
\end{equation}
for measurable, non-negative functions $p$ on $\Omega$. Here
$\beta: \Omega\times (0, +\infty)\rightarrow \R$ is a mapping such that 
$\beta(r, s)$ is a measurable function of $r$ for each $s\in (0, +\infty )$ 
and a strictly convex function of $s$ for each $r\in\Omega$. The definition 
of $\beta$ is extended to $s\leq 0$ by
\begin{equation}
\label{be0}
\beta(r, 0):=\lim_{s\downarrow 0}\beta(r, s), \:\: \:\: \beta(r, s):=+\infty \:\mbox{if $s<0$}.
\end{equation}
No differentiability assumptions are made about $\beta$ but the convenient 
notations $\beta'(r, 0)$ and $\beta'(r, +\infty)$ will be used for the common 
limits of the left and right derivatives of $\beta(r, s)$ by $s$ as 
$s\downarrow 0$ resp. $s\uparrow +\infty$. Note that 
\begin{equation}
\label{bei}
\beta'(r, +\infty) = \lim_{s\uparrow +\infty} \frac{\beta(r, s)}{s}.
\end{equation}
With the understandings \eqref{be0}, the mapping $\beta: \Omega\times
\R\rightarrow (-\infty, +\infty]$ is a convex normal integrand in the sense of
\cite{RockafellarWets1997}, which ensures the
measurability\footnote{Measurability issues will not be entered below. For the
  measurability of functions we deal with, see references in
  \cite{CsiszarMatus2012art} to the book of \cite{RockafellarWets1997}.} of
the function $\beta(r, p( r ))$ in \eqref{H} %below, 
and of similar functions later on, as in \eqref{K} and %$p_\theta$ in 
\eqref{pte}.

%(An obvious necessary condition for $H( p ) < + \infty$ is $p\geq 0$
%$\mu$-a.e.) 
Depending on the choice of $\beta$, $H( p )$ will be relative entropy to $\Pd_0$, some Bregman distance, some $f$-divergence, or some other divergence, as in Section~\ref{sec-applications} below. %In the case $\Gamma$ is some divergence ball centered at $\Pd_0$, as in \eqref{eq-Gamma}, the properties of divergences imply that $$H( p )\geq H( p_0 ) = \Finf = 0$$ for any $p$ with $\int p d\mu = 1$.
%($\Finf$ is defined in \eqref{Fi} below.)
Our general assumption about the relation of $\beta$ and the best guess distribution $\Pd_0$, always satisfied
in the above cases, will be that the minimum of $H(p)$ among probability
densities $p$ is attained for $p_0$, the density of $\Pd_0$;
without any loss of generality, this minimum is supposed to be $0$, thus
\begin{equation}\label{Hnull}
H(p)\ge H(p_0)=0 \:\:\text{whenever}\:\:\int p d\mu=1.
\end{equation}
In addition, we assume that $E_{\Pd_0}(X)= \int X p_0 d\mu$
exists and 
 \begin{equation}\label{Hexp}
% m := \mbox{$\mu$-ess inf}(X) \: <b_0 :=\int X p_0\, d\mu   <\:M := \mbox{$\mu$-ess sup}(X). 
m := \mbox{$\mu$-ess inf}(X) \: <b_0 := E_{\Pd_0}(X)   <\:M := \mbox{$\mu$-ess sup}(X). 
\end{equation}

\paragraph{Relation of the model risk problem and the moment problem}
The distribution model risk \eqref{eq-Delbaenrep} with $\Gamma$ as in \eqref {hbek} is evaluated by solving the worst case problem
\begin{equation}
\label{V}
\inf_{p:\int p d\mu =1, H( p ) \leq k} \int X p d\mu =: V( k )
\end{equation}
and then taking $\MaxLoss=-V(k)$. Our goal is to determine $V(k)$, and also the minimiser (the density of the worst case scenario in $\Gamma$), if $V(k)$ is finite and the minimum in \eqref{V} is attained. If this minimiser exists, it is unique, by strict convexity of $\beta$. 

Problem~\eqref{V} is related to the the moment problem
\begin{equation}\label{F}
\inf_{p:\int p d\mu =1, \int X p d\mu=b} H( p ) =: F(b)
\end{equation} 
in analogy to the relation between problem \eqref{eq-genMaxLoss} and the maximum entropy problem \eqref{eq-MaxEnt} described in Section~\ref{sec-intuition}. Denote
$$\kmax:=\lim_{b\downarrow m }F(b).$$

\begin{proposition}\label{simple}
Supposing 
\begin{equation}\label{k<}
0<k<k_{\max},
\end{equation} 
there exists a unique $b$ with $m<b< b_0$ and
\begin{equation}\label{b<}
F(b)=k,
\end{equation} 
and then the solution to problem \eqref{V} has the value 
\begin{equation}
\label{Vb}
V(k)=b.
\end{equation}
The minimum in \eqref{V} is attained if and only if that in 
\eqref{F} is attained (for this $b$), in which case the same $p$ attains 
both minima.
\end{proposition}
\begin{proof}
As the convex function $F$ attains its minimum $0$ at $b_0$, %is convex
                                %continuous and strictly decreasing in the
                                %interval $(m, t_0)$, with limits $\Flim$ and
                                %$\Finf$, 
the assumption~\eqref{k<} trivially implies the existence of a unique $b$ 
satisfying \eqref{b<}. Moreover, 
then each $t \in (b, b_0)$ satisfies $F(t)< k$, hence there exist functions 
$p$ with $\int p d\mu =1$, $\int X p d\mu = t$ such that $F(t)> k$. This 
proves that $V(k) \leq b$. On the other hand, $F(t) > k$ if $t\in(m,b)$ 
(hence also $F(m) > k$ if $m$ is finite),
%if $t=m$, provided $(1,m)\in\dom\:J$), 
which means that the conditions $\int p d\mu =1$ and  $\int X p d\mu = t$ 
imply $H( p ) \geq F(t) > k$ for each $t \in (- \infty, b)$. %The same holds 
%also when  $\int X p d\mu = -\infty$ (possible only if $m=-\infty$) because 
%then $H( p ) = +\infty$. % by Lemma~\ref{FH}.
%\rem{needs proof} 
Since $\int X p d\mu > -\infty$ if $H( p ) <\infty$, as verified later (Corollary~\ref{cor-2} of
Theorem \ref{second}), this proves that $V(k)\geq t$. The last assertion of the Proposition follows 
obviously. 
\end{proof}

\begin{remark} \label{r-kmax}
The condition \eqref{k<} in Proposition~\ref{simple} covers all interesting 
values of $k$. Indeed, one easily sees that if $k> \kmax$ or $k\ge
\kmax>0$ then $V(k)=m$, while clearly $V(0)=b_0$. This also means that
the functional $H$ can be suitable for assigning model risk only if $\kmax>0$. % Lemma~\ref{Hlim} in the Appendix 
%implies that $k_{max}>0$ always holds if $m$ is finite, and otherwise it  
A necessary and sufficient condition for $\kmax>0$, analogous 
to condition (ii) in Section~\ref{sec-intuition}, will be given
in Corollary~\ref{cor-1} of Theorem~\ref{second}. Note that if $m=-\infty$ then
$\kmax>0$ implies $\kmax=\infty$, in which case each $k>0$ meets
condition \eqref{k<}. 
\end{remark}

%\paragraph*{ }
For technical reasons, it will be convenient to
regard $F(b)$ as the instance $a=1$ of the function
\begin{equation}
\label{J}
J(a, b) := \inf_{p:\int p d\mu =a, \int X p d\mu =b} H( p ), \:\:\:\: (a, b)\in\R^2.
\end{equation}
Problem~\eqref{J} is a special case of minimising convex integral 
functionals under moment constraints, which has an extensive literature. For 
references, see the recent work of \cite{CsiszarMatus2012art}, relied upon 
here also for results that date back much earlier, perhaps under less general 
conditions. %For standard convex analysis tools, our reference is
            %\cite{Rockafellar1970}.
 The results in \cite{CsiszarMatus2012art} will be used (without further
mentioning this) with the choice $\phi:r\rightarrow(1, X( r ))$ of
%pointing out this, applied 
the moment mapping  when the ``value function'' there
%in \cite[eq. (2)]{CsiszarMatus2012art} 
reduces to the function $J$ here. %in \eqref{J}. 
Many results in that reference need a condition called dual 
constraint qualification which, however, always holds in the current setting, 
namely, the set $\Theta$ defined in \eqref{Te} is non-empty (see the passage 
following \eqref{Te}).

The role of the function $\Lambda$ in Section~\ref{sec-intuition} will be played 
by the function 
\begin{equation}
\label{K}
K(\theta_1, \theta_2) := \int \beta^*(r, \theta_1 + \theta_2X( r ))\mu(dr), \: \: \: (\theta_1, \theta_2) \in \R^2,
\end{equation}
where $\beta^*$ is the convex conjugate of $\beta$ with respect to the second
variable,
\begin{equation}\label{betastar}
\beta^*(r, x):=\sup_{s\in \R}\left(xs - \beta(r,s)\right), \: \: \: x\in \R.
\end{equation} 
%Denote the solution of the maximum entropy problem by
%\begin{equation}
%\label{F}
%F(t):=  \inf_{p:\int p d\mu =1, \int X p d\mu=t} H( p ).
%\end{equation}
%and 
%\begin{equation}
%\label{Fi}
%\Finf:= \inf_{t\in (m,M)} F(t).  \:\:\:\: \Flim:= \lim_{t\downarrow m} F(t).
%\end{equation}
The properties of $\beta$ imply that $\beta^*(r, x)$ is 
%measurable in $\omega$, and 
a convex function of $x$ which is finite, non-decreasing, and 
differentiable in the interval $(-\infty, \beta'(r, +\infty))$, see 
\eqref{bei}.
At $x=\beta'(r, +\infty)$, if finite, $\beta^*(r,x)$ may be finite 
or $+\infty$. The derivative $(\beta^*)'(r, x)$ equals zero for 
$x \le\beta'(r, 0)$, is positive for $\beta'(r, 0)< x < \beta'(r, +\infty)$, and grows to $+\infty$ as 
$x\uparrow \beta'(r, +\infty)$. 

The following functions on $\Omega$ will play the role of the 
exponential family, but are parametrised by two variables and need not integrate to $1$:
\begin{equation}
\label{pte}
p_\theta( r ):= (\beta^*)'(r, \theta_1 + \theta_2 X( r )), \:\:\:\:\: \theta=(\theta_1, \theta_2)\in\Theta
\end{equation}
where\footnote{The definition~\eqref{Te} makes sure that the derivative in \eqref{pte} exists for $\mu$-a.e. $r\in\Omega$ if $(\theta_1, \theta_2)\in\Theta$. For all other $r\in\Omega$, if any, one may set by definition $p_{\theta_1, \theta_2}=0$.} 
\begin{equation}
\label{Te}
\Theta:=\left\{ \theta: K(\theta_1, \theta_2)< + \infty, \; \; \theta_1 + \theta_2 X( r ) < \beta'(r, +\infty) \: \mu\mbox{-a.e.}\right\}.
\end{equation}
The properties of $\beta^*$ stated above imply for any
$(\theta_1, \theta_2)$ in the effective domain $\dom\, K:=\{(\theta_1,
\theta_2):K(\theta_1, \theta_2)<+\infty\}$ of $K$ that $(\overline{\theta_1},
\theta_2)\in\Theta$ for each $\overline{\theta_1} < \theta_1$. In particular,
$\Theta$ contains the interior of $\dom\, K$. If $\beta'(r,+\infty)=+\infty$ $\mu$-a.e. then $\Theta=\dom\, K$.
As verified later, see Remark 4, the default density $p_0$ is equal to $p_{(\theta_0,0)}$ for some $\theta_0$ with
$(\theta_0,0)\in\Theta$. 

The function $K$ is equal to the convex conjugate of $J$:
\begin{equation}
\label{KJ}
K(\theta_1, \theta_2) = J^*(\theta_1, \theta_2) := \sup_{(a, b)\in\R^2}\left(\theta_1 a+ \theta_2 b - J(a, b)\right), \end{equation}
%This follows from $H(\Pd_0)=0$ and $E_{\Pd_0}(X) > - \infty$ by
see \cite[Theorem 1.1]{CsiszarMatus2012art}. In particular, $K$ is a lower
semicontinuous proper\footnote{I.e., it never equals $-\infty$ and is not
identically $+\infty$.} convex function. Also, $K$ is differentiable in the
interior of $\dom\, K$, and
\begin{equation}
\label{grad}
\nabla K(\theta) = \left( \int p_{\theta} d\mu, \int X p_{\theta} d\mu\right), 
\:\:\:\theta=(\theta_1, \theta_2) \in \mbox{int\,dom} \, K,
\end{equation}
see \cite[Corollary 3.8]{CsiszarMatus2012art}.

%The function $K$ is relevant for our purposes because it admits a manageable 
%representation of the function $J$ in \eqref{J}, at least in the interior of 
%its effective domain. Namely, $J^*=K$ implies 
\paragraph{Main results} We calculate $b$ satisfying~\eqref{b<}, which by \eqref{Vb} amounts to solving problem \eqref{V}, by evaluating instead of $J$ the function $K^*$, using the identity $J^*=K$ which implies (\citet[Theorem 12.2]{Rockafellar1970})
\begin{equation}
\label{JK}
J(a, b) = K^*(a, b), \:\:\:\: (a, b) \in \intdom \, J.
\end{equation}
$K^*$ is the convex conjugate of $K$,
\begin{equation}
\label{Kstar}
K^*(a, b):=\sup_{(\theta_1, \theta_2)\in\R^2} \left(\theta_1 a + \theta_2 b - K(\theta_1, \theta_2)\right),
\:\:\:\: (a, b) \in \R^2,
\end{equation}
and the interior of the effective
domain of $J$ is, by \citet[Lemma 6.6]{CsiszarMatus2012art} 
\begin{equation}\label{idoJ}
\intdom\, J =\{(a,b): a>0,\; am<b<aM\}.
\end{equation}
Proposition \ref{simple} and \eqref{JK}, \eqref{idoJ} imply for $0<k<k_{\max}$ the analogue of
Remark~\ref{r-Lambda*}: 
%at the end of Section~\ref{section-RelEnt-pathological}, 
A (unique) $b$ satisfies
\begin{equation}\label{conjeq}
K^*(1,b)=k\quad\mbox{and}\quad b<b_0=E_{\Pd_0}(X),
\end{equation}
and then $V(k)=b$. This already provides a recipe for
computing $V(k)$. In regular cases, a more explicit solution
is available, based on the following
key result about Problem~\eqref{J}, 
see \cite[Lemma 4.4, Lemma 4.10]{CsiszarMatus2012art}:

\begin{lemma}\label{l-basic}
If $\theta=(\theta_1,\theta_2)\in\Theta$
satisfies 
\begin{equation}\label{e-moment}
\int p_{\theta}\,d\mu = a,\quad \int Xp_{\theta}\,d\mu =b
\end{equation}
then it attains the maximum in \eqref{Kstar}. Moreover,
in case $(a, b) \in \intdom\; J$, the existence of 
$\theta\in\Theta$ satisfying~\eqref{e-moment}
%$$\int p_{\theta}\,d\mu = a,\quad \int Xp_{\theta}\,d\mu =b $$
 is necessary and 
sufficient for the attainment of the minimum in \eqref{J}, and then 
$p=p_{\theta}$ is the (unique) minimiser.
\end{lemma}

\begin{theorem} \label{main}
Assuming \eqref{Hnull}, \eqref{Hexp}, \eqref{k<}, if
\begin{equation}
\label{dK}
\thbar_2<0, \quad \int p_{\thbar}\,d\mu = 1, %\label{dK1}\\
\quad\thbar_1 + \thbar_2 \int Xp_{\thbar}\,d\mu - K(\thbar) = k %\label{dK2} 
\end{equation} 
for some $\thbar=(\thbar_1,\thbar_2)\in\Theta$  then the value of the inf in \eqref{V} is
\begin{equation}
\label{MRk-solution}
V(k) = \int Xp_{\thbar}\,d\mu. 
\end{equation}
%and $p_{\theta}$ is the worst case density, attaining the minimum in
%\eqref{V}. 
Essential smoothness\footnote{A lower semicontinuous proper convex function is essentially 
smooth if its effective domain has nonempty interior, the function is
differentiable there, and at non-interior points of the effective domain 
the directional derivatives in directions towards the
interior are $-\infty$. The latter trivially holds if the effective 
domain is open.} of $K$ is a sufficient condition for
the existence of such $\thbar$. Further, a necessary and sufficient condition for $p$ to attain the minimum 
in~\eqref{V} is $p=p_{\thbar}$ for the $\thbar\in\Theta$ satisfying~\eqref{dK}.
%is necessary and sufficient for $p$ to attain the minimum in~\eqref{V}.
\end{theorem}

% As a corrolary, we get a procedure to determine the solution of Problem~\eqref{V} by solving two equations in two real variables.

\begin{corollary}\label{r-diffeq}  
If the equations 
\begin{equation}\label{diffeq}
\frac{\partial}{\partial\theta_1} K(\theta) = 1,\quad \theta_1 + 
\theta_2 \frac{\partial}{\partial\theta_2} K(\theta) - K(\theta) = k
\end{equation} 
have a solution $\thbar=(\thbar_1, \thbar_2)\in \intdom \, K$
with $\thbar_2 < 0$ then $\thbar$ satisfies~\eqref{dK} and the solution to Problem \eqref{V} equals 
\begin{equation}
V(k) = \left. \frac{\partial K(\theta)}{\partial \theta_2} 
\right\vert_{\theta=\thbar}\,.
\end{equation}
\end{corollary}
The Corollary follows from the Theorem because, for $\thbar\in\intdom \, K$, the equations in~\eqref{dK} are
equivalent to those in~\eqref{diffeq}, by~\eqref{grad}. However, if $K$ is not essentially smooth, 
$\thbar\in\intdom \, K $ is not a necessary condition for  \eqref{dK}.

\begin{proof}
By Lemma~\ref{l-basic}, if $\theta=(\theta_1,\theta_2)\in\Theta$
satisfies 
\begin{equation}\label{momenteq}
\int p_{\theta}\,d\mu = 1, \int Xp_{\theta}\,d\mu =b
\end{equation}
then it attains the maximum in \eqref{Kstar}. It follows, using \eqref{JK}, that
\eqref{dK} implies for $b:=\int Xp_{\thbar}\,d\mu$, if it satisfies
$m<b<M$, that
 %If \eqref{dK} holds then, by the assertions sent forward about
%\eqref{Kstar},\eqref{JK2} applied to $a=1,\;b=\int Xp_{\theta}\,d\mu$, we have
\begin{equation}\label{F=}
F(b)=J(1,b)=\thbar_1+\thbar_2 b-K(\thbar)=k.
\end{equation}
Due to Proposition~\ref{simple}, to prove \eqref{MRk-solution}  
it remains to show that $m<b<b_0$. Clearly, $k<k_{\max}$ implies $m<b$. 
Further, \eqref{Kstar} and \eqref{F=} imply
\begin{equation}\label{F>}
F(t)=K^*(1,t)\ge \thbar_1+\thbar_2 t-K(\thbar_1,\thbar_2)=F(b)+\thbar_2(t-b),
\;\; t\in(m,M). 
\end{equation}
Since $\thbar_2<0$, this shows that $F(t)>F(b)$ if $t\in(m,b_0)$,
completing the proof of \eqref{MRk-solution}. %The assertion that
%$p_{\max}$ is the worst case density follows from Proposition~\ref{simple} 
%and the passage preceding Theorem~\ref{main}.

Suppose next that $K$ is essentially smooth. Then to $b$ in \eqref{b<} 
%Proposition \ref{simple}. Then 
there exists $\thbar\in\intdom\,K$ with
\begin{equation}\label{=nabla}
(1,b)=\nabla K(\thbar), 
\end{equation}
%for some $\theta\in\intdom\,K$, 
because $(1,b)\in \intdom\,J$ and
the gradient vectors of the essentially smooth $K$
cover $\intdom\,K^*=\intdom\,J$, see \citet[Corollary 26.4.1]{Rockafellar1970}. Clearly, \eqref{=nabla} implies that $\thbar$ attains the 
maximum in \eqref{Kstar}, hence it satisfies \eqref{F=}. 
 %thus $F(b)=J(1,b)=\theta_1+\theta_2 b-K(\theta)$.
%Recalling \eqref{grad}, the last 
This means by \eqref{=nabla} that
 $\thbar$ satisfies the equations in \eqref{diffeq}, equivalent to those in
\eqref{dK}.  It remains to show that
$\thbar_2<0$, but this follows from \eqref{F>} applied to $t=b_0$.

Finally, the last assertion of Theorem \ref{main} follows from 
Proposition \ref{simple} and Lemma~\ref{l-basic}.
\end{proof}

\paragraph{Conditions for $\kmax > 0$.}
In Proposition~\ref{simple} and Theorem~\ref{main} the condition $\kmax > 0$ has been assumed. In this subsection we give a 
necessary and sufficient condition for this to hold. We begin with a remark.

\begin{remark}\label{r-p0} A simpler instance of 
\cite[Lemma 4.10]{CsiszarMatus2012art} than Lemma~\ref{l-basic},
namely with the constant mapping $r\rightarrow 1$ taken
for the moment mapping $\phi$, gives the following: the 
 necessary and 
sufficient condition for $p$ to minimise $H(p)$ subject to 
$\int p\,d\mu = a$ ($a>0$) is that $p(r)=(\beta^*)'(r, \theta)$
for some $\theta\in\R$ with $\beta^*(r, \theta)$ $\mu$-integrable, and then
the minimum is equal to $a\theta-\int \beta^*(r, \theta)\,d\mu(r)$. 
This establishes the claim that the default density
$p_0$, minimising $H(p)$ subject to $\int p\,d\mu = 1$, equals
$p_{(\theta_0,0)}$ for some $\theta_0$ with $(\theta_0,0)\in\Theta$; 
this $\theta_0$ also satisfies
$\theta_0-\int \beta^*(r, \theta_0)\,d\mu(r)=H(p_0)=0$.
\end{remark}

\begin{theorem}\label{second}
%\begin{lemma}\label{l:Hlim}\rem{move to Appendix}
Assuming \eqref{Hnull},~\eqref{Hexp}, for $b<b_0$ we have $F(b)>0$ 
\mbox{if and only if}
\begin{equation}\label{poscon}
\mbox{there exists}\:\: \theta=(\theta_1,\theta_2)\in
\dom\,K\:\:\mbox{with}\:\:
\theta_2<0.
\end{equation}
\end{theorem}
\begin{proof} To prove the necessity of~\eqref{poscon}, we may assume
$m<b<b_0$. Then $(1,b)\in\intdom\,J$, see~\eqref{idoJ}, hence the convex
function $J$ has nonempty subgradient at $(1,b)$ ~\cite[Theorem 23.4]{Rockafellar1970}.
As $J^*=K$, if $\thbar=(\thbar_1,\thbar_2)$ belongs to that subgradient then
\begin{equation}\label{e-attain}
F(b)=J(1,b)=\thbar_1+\thbar_2 b-K(\thbar)
\end{equation}
by~\cite[Theorem 23.5]{Rockafellar1970}, %and~\eqref{e-attain} 
which implies as in the proof of 
Theorem~\ref{main} that this $\thbar$ also satisfies~\eqref{F>}. In turn,
~\eqref{F>} with $t=b_0$ implies that $\thbar_2\le 0$, with the strict 
inequality if $F(b)>0$. This proves the necessity of ~\eqref{poscon}.

For sufficiency, suppose that $F(b)=0$ for some 
$b\neq b_0$, $m<b<M$. By Remark~\ref{r-p0}, then 
 $F(b)=0 =\theta_0-K(\theta_0, 0)$, hence $\theta^*:=(\theta_0, 0)$ is a
maximiser of $g(\theta):= \theta_1+\theta_2 b-K(\theta)$, see \eqref{JK},
\eqref{Kstar}. It follows that for no $\bar{\theta}\in\dom\,K$ can the 
directional derivative $g'(\theta^*;\bar{\theta}-\theta^*)$ be positive. By
\cite[Lemma 3.6, Remark 3.7]{CsiszarMatus2012art}, this directional
derivative is equal to 
$$ (\bar{\theta}_1-\theta_0)+\bar{\theta}_2 b-\int (\bar{\theta}_1-\theta_0+
\bar{\theta}_2 X)p_{\theta^*}\,d\mu=\bar{\theta}_2 (b-b_0).$$
Thus, the existence of $\bar{\theta}\in\dom\,K$ with 
$\bar{\theta}_2<0$ rules out $b<b_0$, proving the sufficiency part of the
Theorem.  
\end{proof}

%following condition is necessary and sufficient $F$ in \eqref{F} either 
%strictly decreases or identically vanishes in the interval $(m,b_0)$,
%according as $\dom\,K$ does or does not contain $\theta=(\theta_1,\theta_2)$ 
%with $\theta_2<0$. If $m$ is finite then always the first contingency holds.
%\end{lemma}
\begin{corollary} \label{cor-1}
Condition~\eqref{poscon} is necessary and sufficient for $k_{\max}>0$. 
Sufficient conditions are the finiteness of $m$ or the essential smoothness 
of $K$. 
\end{corollary}
\begin{proof}
 If $m$ is finite then each $\theta_2<0$ 
satisfies condition~\eqref{poscon} with some $\theta_1$. Indeed, since
$\theta_1+\theta_2 X\le\theta_1+\theta_2 m$ $\mu$-a.e.,
if the right hand side is less than $\theta_0$ in Remark
\ref{r-p0} then $(\theta_1,\theta_2)\in\dom\,K$. If $K$ is essentially smooth 
then condition~\eqref{poscon} holds because $\intdom\,K$ contains
 $\theta^*=(\theta_0,0)$. Indeed, otherwise the directional derivatives of $K$
at $\theta^* $  in directions towards interior points were equal to $-\infty$,
and $\theta^*$ could not maximize $\theta_1+\theta_2 b_0-K(\theta)$.
\end{proof}
\begin{corollary} \label{cor-2}
If $k_{\max}>0$ then $\int p d\mu=1,\:H(p)<+\infty$ imply
\mbox{$\int X p d\mu>-\infty$.}
\end{corollary}
\begin{proof}
Substitute in the Fenchel inequality  $ xs\le\beta(r,s)+   \beta^*(r,x) $
(a consequence of \eqref{betastar})
$x:=\theta_1+\theta_2X(r)$, $s:=p(r)$ and integrate. It follows that
if $(\theta_1,\theta_2)\in\dom\,K$ and $p$ satisfies the hypotheses then
$$ \theta_1+\theta_2 \int Xpd\mu\le H(p)+K(\theta_1,\theta_2)<+\infty.$$
Taking $(\theta_1,\theta_2)$ as in~\eqref{poscon}, the assertion follows. 
\end{proof}

\section{MaxLoss over Bregman balls and $f$-divergence balls} \label{sec-applications}
We now come back to the more specific choices \eqref{eq-Gamma}, where $\Gamma$
is a ball of distributions in terms of some divergence $D$, centered at some 
$\Pd_0$.  

\paragraph{Relative entropy balls}
Let us briefly check how the unified framework leads,
in the special case of relative entropy balls, 
%. In this case the
%function $\beta$ is given by $\beta(r,
%s):=f(s) := s\log s - s +1$ and
%Theorem~\ref{main} reduces 
to the results of \citet[Theorem 1]{BreuerCsiszar2012} reported in
Section~\ref{sec-intuition}.

Set $\mu=\Pd_0$ and take $\beta(r, s):=
f(s) := s\log s - s +1$. Then
 %when we choose $\Pd_0$ for $\mu$. % The derivative $\beta'(r, s)=\log s$, so we have log utility. 
%The function 
$\beta^*(r, x)=f^*(x)=\exp(x)-1$ and ${\beta^*}'(r, x)=(f^*)'(x)=\exp(x)$. 
Hence, using \eqref{eq-def-Lambda} and \eqref{K}, 
$$ K(\theta_1,\theta_2)=\int(\exp(\theta_1+\theta_2 X)-1)\,
d\Pd_0=\exp(\theta_1+\Lambda(\theta_2))-1 $$
and $\Theta=\dom\;K=\R\times\dom\;\Lambda.$ 
The %generalised exponential 
functions $p_\theta$, $\theta\in\Theta$ of \eqref{pte} %consists
                                %of the %non-negative functions 
are of form
$\exp(\theta_1 + \theta_2 X( r ))$, and integrate to $1$ if and only if
$\theta_1=-\Lambda(\theta_2)$. Then $p_\theta$ is the $\Pd_0$-density of 
$\Pd(\theta_2 )$ in the exponential family~\eqref{eq-wcs}.
%$K(\theta)$ from \eqref{K} is
%$$K(\theta)=e^{\theta_1} e^{\Lambda(\theta_2)}-1.$$

The first equation in 
\eqref{dK} requires 
%reads $\exp(\thbar_1)\exp(\Lambda(\thbar_2))=1$, which implies $\thbar_1= -
%\Lambda(\thbar_2)$ as the condition for 
$p_\theta$ to be a density, thus $\theta_1= -\Lambda(\theta_2)$. Then
$\int Xp_\theta d\Pd_0=\Lambda'(\theta_2)$, see~\eqref{eq-Lambda'}, and the
%; this density equals the $\Pd(\theta_2)$ given in \eqref{eq-wcs}. Note that
%$K(\thbar)=0$. The 
second equation in %\eqref{diffeq} and in 
\eqref{dK} reads 
$-  \Lambda(\theta_2)    +\theta_2 \Lambda'(\theta_2)=k$, which is 
\eqref{eq-def-theta}. Thus Theorem~\ref{main} gives the result in 
Section~\ref{sec-intuition} that if ~\eqref{eq-Lambda'} has a negative
solution $\thbar$ then $V(k)=\Lambda'(\thbar)$, a worst case scenario
exists, and its density is $p_{\thbar}$.

%The optimal value $V(k)$ equals $\partial/\partial\theta_2 K(\thbar)$ by \eqref%{MRk-solution}, which equals $\Lambda'(\thbar_2)$, as in equation~\eqref{eq-Lam%bda'}.

\paragraph{$f$-divergence balls}
Setting  $\mu=\Pd_0$ again, take now any autonomous integrand for $\beta$
given by a convex function $f$ as in Section~\ref{sec-relative entropy}, and
let
\begin{equation}\label{eq-Hf}
H(p):=\int f(p)\,d\Pd_0.
\end{equation}
Then the set $\Gamma$ of distributions given by~\eqref{hbek} is equal to the
$f$-divergence ball $\{\Pd:D_f(\Pd||\Pd_0)\le k\}$ if $f$ is cofinite,
while if $f'(+\infty):=\lim_{s\to\infty}f(s)/s$ is finite, $\Gamma$
is a proper subset of that ball. We will focus on $\Gamma$ defined by 
~\eqref{hbek} anyway.

If $f$ is not cofinite then $f^*(x)=+\infty$ for $x>f'(+\infty)$, hence
$$ K(\theta_1,\theta_2)=\int f^*(\theta_1+\theta_2 X)\, d\Pd_0$$
is infinite when $\theta_2<0$, unless $m:=\ess\:\, \inf(X)$ is finite.
By Corollary~\ref{cor-2} of Theorem~\ref{second}, this means that the 
functional~\eqref{eq-Hf} can be adequate for assigning model risk only if
$f$ is cofinite or if $X$ is essentially bounded below. In the latter case,
$(\theta_1,\theta_2)$ with $\theta_2<0$ belongs to $\intdom K$ if and only if
$\theta_1+\theta_2 m<f'(+\infty)$.

The most poular $f$-divergences are the \emph{power divergences}, defined by
$$
f_{\alpha}(s):=[s^\alpha -\alpha(s-1) -1]/[\alpha(\alpha-1)],\quad \alpha\in\R.
$$
Formally, $f_{\alpha}$ is undefined if $\alpha=0$ or $\alpha=1$, but 
the definition is commonly extended by limiting, thus
$$ f_0(s):=\log s+s-1,\quad\quad f_1(s):=s\log s -s+1. $$
This means that also $D_{f_0}(\Pd||\Pd_0)=I(\Pd_0||\Pd)$ and
$D_{f_1}(\Pd||\Pd_0)=I(\Pd||\Pd_0)$ are regarded as power divergences. Note
that the function $f_{\alpha}$ is cofinite if and only if $\alpha\ge 1$, and 
$f'_{\alpha}(+\infty)=1/(1-\alpha)$ if $\alpha< 1$.

Let us determine the family of functions 
\begin{equation}\label{pow}
 p_{\theta}(r)=(f_{\alpha}^*)'(\theta_1+\theta_2 X(r)),\quad
\theta=(\theta_1,\theta_2)\in\Theta
\end{equation}
that contains the worst case densities in power divergence balls, more
exactly, in ~\eqref{hbek} with $f=f_{\alpha}$. Since
$f'_\alpha(s) = [s^{\alpha-1}-1]/(\alpha-1)$ grows from
$-\infty$ to $1/(1-\alpha)$ if $\alpha<1$ or from $1/(1-\alpha)$ to $+\infty$
if $\alpha>1$, as $s$ runs over $(0,+\infty)$. In the interval 
$(-\infty,1/(1-\alpha))$ or $(1/(1-\alpha),+\infty)$, respectively,
$(f_\alpha^*)'$ is the inverse function of $f'_\alpha$, thus
\begin{eqnarray*}
(f_\alpha^*)'(x)=[x(\alpha-1)+1]^{1/(\alpha-1)}&\mbox{if}& \alpha<1,\:x<1/(1-\alpha)\\
&\mbox{or}&\alpha>1,\:x>1/(1-\alpha).
\end{eqnarray*}
Clearly, $(f_\alpha^*)'(x)$ does not exist if $\alpha<1$ and
$x\ge 1/(1-\alpha)$,
while if $\alpha>1$ and $x\le 1/(1-\alpha)$ then $(f_\alpha^*)'(x)=0.$
This gives a simple formula for the functions $ p_{\theta}$ in~\eqref{pow}. Unlike for the relative entropy case, however, no
explicit condition is available for $\int  p_{\theta} d\Pd_0=1$, and the
two equations in Theorem~\ref{main} cannot be reduced to one.

%If $D$ in \eqref{eq-Gamma} is chosen to be the $f$-divergence of \eqref{fdi}, this amounts to choosing the autonomous integrand 
%(not depending on $r\in\Omega$) 
%$\beta(r,s)=f(s)$ and $\mu=\Pd_0$. Then $H( p )=\int f( p ) d\Pd_0$ and $\Gamma$ is the set of probability measures $\Pd \ll \Pd_0$ with
%$D_f(\Pd || \Pd_0) \leq k$. 

%More generally, if $\Pd_0$ is different from $\mu$, we have $$D_f(p,p_0) = \int_\Omega p_0( r ) f\left(\frac{p( r )}{p_0( r )}\right) d\mu( r ).$$ Thus $\beta(r,s):=q_0( r )f\left(s/q_0( r )\right)$. The convex conjugate of $\beta$ with respect to the variable $s$ equals $\beta^*(r,x) = q_0( r )f^*(x)$. The generalised exponential family is formed by the functions
%$p_\theta( r )=q_0( r ){f^*}' (\theta_1+\theta_2 X( r )).$
%Probability densities in the exponential family satisfy $\int_\Omega q_0( r ){f^*}' (\theta_1+\theta_2 X( r )) d\mu( r ) = 1$.
%The function $K$ is given by
%$$K(\theta)=\int_\Omega q_0( r )f^*(\theta_1+\theta_2 X( r )) d\mu( r ).$$
%The worst case density is $q_0( r ){f^*}' (\thbar_1+\thbar_2 X( r ))$, where $\thbar_1, \thbar_2$ are determined as solutions to eqs.~\eqref{diffeq}.

% {\bf Questions:} Under which condition on $f$ is $K$ essentially smooth? Can we get an explicit expression for $\Flim$?

\paragraph{Bregman balls}
In the special case $\mu=\Pd_0$, the Bregman distance~\eqref{Bf} reduces to 
$f$-divergence: If $f$ is a non-negative convex function with $f(1)=0$ and 
differentiable at $s=1$ then $\Delta_f(s,1)=f(s)$, consequently 
%\begin{equation}\label{fB}
$$B_{f,\Qd}(\Pd,  \Qd) = D_f(\Pd || \Qd) \:\: \mbox{for $\Pd \ll \Qd$}.$$
%\end{equation}
Hence, in this subsection, $\mu$ is taken different from $\Pd_0$; for
simplicity, $f$ is assumed differentiable. To obtain for
%In the general case, where $\mu$ might differ from $\Pd_0$, we get 
$D$ in \eqref{eq-Gamma} resp. $H$ in \eqref{H} %to be 
the Bregman distance $B_{f,\mu}$ of \eqref{Bfp}, %if 
we choose the non-autonomous integrand 
$$\beta(r, s)=  f(s) - f(p_0( r ))
-f'(p_0( r ))(s- p_0( r )).$$ To make sure that this meets the assumptions
on $\beta$, in case $f'(0)= -\infty$ we assume that the default 
density $ p_0$ is $\mu$-a.e. positive; this assumption is not needed
if  $f'(0)> -\infty$.

By \citet[Lemma 2.6]{CsiszarMatus2012art}, the convex conjugate of $\beta$
with respect to $s$ equals $$\beta^*(r,x) = f^*(x + f'(p_0( r )))- f^*( f'(p_0( r ))).$$
The function $K$ from \eqref{K} equals 
$$K(\theta):= \int_\Omega \left[ f^*(\theta_1 + \theta_2 X( r ) + f'(p_0( r )))- f^*( f'(p_0( r ))\right]d\mu( r ). $$ 
The %exponential 
family $\{p_\theta( r ): \theta\in\Theta\}$ is formed by the (non-negative) 
functions %of the form ${\beta^*}'(\theta_1+\theta_2 X( r )))$, which by 
%\citet[Lemma 2.6]{CsiszarMatus2012art} equals
%\begin{equation}
%\label{def-qnew}
$$p_\theta( r ) = {\beta^*}'(\theta_1+\theta_2 X( r ))= {f^*}'[\theta_1 + \theta_2 X( r ) + f'(p_0( r ))].$$
%\end{equation}

Note that while the case of Bregman balls is covered by
our general results, it is not 
apparent that the current special form of $\beta$ would substantially 
simplify their application. 

% {\bf Questions:} Under which condition on $f$ is $K$ essentially smooth? Can we get an explicit expression for $\Flim$?

%\paragraph{Worst case over power divergence balls}
%$H( p )$ equals the power divergence of $\Pd$ from $\Pd_0$ if we choose $\mu=\Pd_0$ and
%$\beta(r, s) = f_\alpha(s) = [s^\alpha -\alpha(s-1) -1]/[\alpha(\alpha-1)]$, for $\alpha> 0$. Then 
%$\beta'(r, s) = f'_\alpha(s) = [s^{\alpha-1}]/(\alpha-1)$ and 
%$$(\beta^*)'(r, x)= (f_\alpha^*)' = \left( x (\alpha -1) + 1\right)^{1/(\alpha-1)}.$$ In this equation, if $\alpha < 1$ it is required that $x <  1/(1-\alpha)$  and if $\alpha > 1$ one sets $(f_\alpha^*)'(x):= 0$. The exponential family is given by
%$$p_\theta( r ) = \left[(\theta_1 + \theta_2 X( r ))(\alpha-1) + 1 \right] ^{1/(\alpha-1)}$$
%In case $\alpha < 1$ this is well-defined if $\theta_1 + \theta_2 X( r )\leq 1/(1-\alpha)$, so that $\Theta=\{\theta: \theta_1 + \theta_2 m \leq 1/(1-\alpha), K(\theta) <\infty \}$. In case $\alpha > 1$, $p_\theta( r )$ has the same form, but it equals zero if $r < 1/(1-\alpha)$. 
%In case $\alpha= 0$, $p_\theta( r )  = 1/(1-\theta_1 - \theta_2X( r ))$.

\section{Evaluation of divergence preferences}
\label{sec-divergence preferences}
Finally, we briefly address divergence preferences, i.e., the
problem~\eqref{DP} which, in the framework of Section~\ref{sec-CIF}, is 
simpler than the minimization of $H(p)$ over the set~\eqref{hbek}.
%The techniques of Section~\ref{sec-CIF} can be applied directly to evaluate explicilty the divergence preferences \eqref{DP}.
%This 
Divergence preferences include as special case the multiplier preferences of 
\citet{HansenSargent2001}, when we choose the relative entropy~$I$ for $D$. 
\citet{MaccheroniMarinacciRustichini2006} choose for $D$ the more general 
weighted $f$-divergences 
\begin{equation}\label{Dwf}
D_f^w(\Pd,\Pd_0):=  \left\{ \begin{array}{ll}
\int_\Omega w( r ) f\left(\frac{d\Pd}{d\Pd_0}( r )\right) d\Pd_0 ( r ) & \mbox{if $\Pd \ll \Pd_0$,} \\
 +  \infty  & \mbox{otherwise,}
 \end{array} 
 \right.
\end{equation}
where $w$ is a normalised, non-negative weight function. 

Below, more generally, the role of $D$ is given to any convex functional
as in~\eqref{H}. Introducing a new convex integrand and intergal functional by
$$ \tilde{\beta}(r,s):= X(r)s+\lambda \beta(r,s),\quad \tilde{H}(p):=\int
\tilde{\beta}(r,p(r)) d\mu(r), $$
(where $\lambda >0$ is fixed), we can write
\begin{equation}\label{W1}
W:=\inf_{p:\int p d\mu=1}\,\left[\int Xp d\mu+\lambda H(p)\right]=
    \inf_{p:\int p d\mu=1}\,\tilde{H}(p).
\end{equation}
Thus, the problem is to minimize the functional $\tilde{H}(p)$ under the
single constraint $\int p d\mu=1$.

In analogy to~\eqref{J}, consider
$$ \tilde{J}(a):=\inf_{p:\int p d\mu=a}\,\tilde{H}(p), \quad a\in\R.$$
Note that $\tilde{\beta}$ meets the basic assumptions on $\beta$ (though
\eqref{Hnull} does not hold for $\tilde{H}$), and that
$$ (\tilde{\beta})^*(r,x)=\sup_s \left[xs-X(r)s-\lambda\beta(r,s)\right]=
   \lambda\beta^*\left(r,\frac{x-X(r)}{\lambda}\right).$$
It follows by~\cite[Theorem 1.1]{CsiszarMatus2012art} that the convex
conjugate of $\tilde{J}$ equals
$$\tilde{K}(\theta):=\int (\tilde{\beta})^*(r,\theta) d\mu(r)=
   \lambda \int \beta^* \left(r,\frac{\theta-X(r)}{\lambda}\right) d\mu(r),
   \quad \theta\in\R, $$
or, with the notation~\eqref{K},
$$\tilde{J}^*(\theta)=
 \tilde{K}(\theta)= \lambda K(\frac{\theta}{\lambda},- \frac{1}{\lambda}),
\quad \theta\in\R.$$
As the interior of $\dom\;\tilde{J}$ is $(0,+\infty)$, it follows that 
$\tilde{J}(a)=\tilde{K}^*(a)$ for each $a>0.$ In particular,
\begin{eqnarray}
W=\tilde{J}(1)=\tilde{K}^*(1)&=&\sup_{\theta\in\R}\,(\theta-\tilde{K}(\theta))
=\sup_{\theta\in\R}\left[\theta-\lambda K\left(\frac{\theta}{\lambda},- \frac{1}{\lambda}\right)\right] \label{Ktilde}\\ 
&=&\lambda\sup_{\theta_1\in\R}\left[\theta_1- K\left(\theta_1,- \frac{1}{\lambda}\right)\right]. \nonumber
\end{eqnarray}

\begin{proposition} \label{prop-DP}
The necessary and sufficient condition for $W>-\infty$ in~\eqref{W1} 
is the existence of $\theta_1\in\R$ with
\begin{equation}\label{nesu}
(\theta_1,-1/\lambda)\in \dom \; K,
\end{equation}
and then 
\begin{equation}\label{W3} 
W=\lambda \sup_{\theta_1}\left[\theta_1-K(\theta_1,-1/\lambda)\right].
\end{equation}
If for some $\theta=(\theta_1, -1/\lambda)$
 as in~\eqref{nesu} 
the function $p_{\theta}$ in~\eqref{pte} 
has integral equal to one, then
 $\theta_1$ attains the maximum in~\eqref{W3}, and 
$p=p_{\theta}$ attains the minimum in~\eqref{W1}. Otherwise, among the
numbers $\theta_1$ satisfying~\eqref{nesu} there exists
a largest one $\theta_{1\max}$, and $ p_{\theta}$ with 
$ \theta=(\theta_{1\max}, -1/\lambda)$ has integral less than one; 
then $\theta_1=\theta_{1\max}$ attains the maximum in~\eqref{W3}. 
%Assume that $X$ is not $\mu$-a.e. constant and that the function $\Kbar$ of \eqref{Kbar} is essentially smooth. Then there exists a $\thbar_1\in\intdom \Kbar$ solving the equation \eqref{dKbar1} resp. \eqref{ptnorm}. The worst case distribution for problem \eqref{DP} is the distribution $\pbar_{\thbar_1}$ of \eqref{pt2} and the value $W$ assigned by the decision maker  equals \eqref{W}.
\end{proposition} 

\begin{proof}
Clearly, $W=\tilde{J}(1)>-\infty$ if and only if $\tilde{J}$ never equals
$-\infty$, thus
its conjugate $\tilde{K}$ is not identically $+\infty$; by the formula for
$\tilde{K}$, this proves the first assertion. 
The second assertion follows from \eqref{Ktilde}.
 %recalling that the definition 
%$\tilde{K}^*(a)=\sup(\theta a-\tilde{K}(\theta)$ the supremum is actually
%a maximum if $a\in\intdom\tilde{K}^*.$ 
As the supremum in~\eqref{W3} is the same as the supremum defining 
$\tilde{K}^*(1)$ in \eqref{Ktilde} (with $\theta/\lambda$ substituted by
$\theta_1$), the next assertion follows from the simple instance of 
\cite[Lemma 4.10]{CsiszarMatus2012art} used in Remark~\ref{r-p0} 
(note that the function $(\beta^*)'(r,\theta)$ there, replacing
$\beta$ by $\tilde{\beta}$ and $\theta$ by $\theta_1\lambda$, gives
the function $p_{\theta}$ in the Proposition).
For the last assertion, recall that the maximum in the definition of
$\tilde{K}^*(1)$, and therefore in~\eqref{W3}, is always attained,
because $a=1$ is in the interior
of $\dom\,\tilde{K}^*$ (as in Remark~\ref{r-p0}). 
Then the (left) derivative by $\theta_1$ of $K(\theta_1,-1/\lambda)$
at the maximiser, say $\theta_1^*$,  has to be $\le 1$, and the strict 
inequality
can hold only if $\theta^*=\theta_{1\max}$. As the mentioned derivative
equals the integral of $p_{\theta^*}$ with  
$\theta^*=(\theta_1^*,-1/\lambda)$, this completes the proof.
\end{proof}

\paragraph{Evaluation of multiplier preferences}
As an example apply Proposition~\ref{prop-DP} to reproduce a result of \citet{HansenSargent2001}. We evaluate the objective function of an agent with multiplier preferences \eqref{DP} choosing for $D$ the relative entropy. This corresponds to the choice 
$\beta(r,s) = s \log s - s +1$, and $\mu=\Pd_0$. %In this case $f'(s) = \log
                                %s$, $(f')^{-1} (\alpha) = \exp(\alpha)$, and
                                %$(f^*)'(x) = \exp(x)$. Thus the worst case
                                %distribution of \eqref{pt2} has the density
In this case, the condition for $W>-\infty$ in Proposition~\ref{prop-DP}
becomes $-1/\lambda\in \dom\;\Lambda$. Under that condition, the function
$p_{\theta}$ with 
$$\theta=(\thbar_1   , -1/\lambda),\quad\thbar_1=-\Lambda(-\frac{1}{\lambda})$$
has integral equal to one,
hence the Proposition gives that this $p_{\theta}$, namely 
%$$\pbar_{\thbar_1} ( r ) = e^{(\thbar_1 - X( r ))/\lambda},$$
%which by \eqref{eq-wcs} equals
the member $\Pd(-1/\lambda)$ of the exponential family with parameter 
value $-1/\lambda$, attains the minimum in the definition~\eqref{W1} of $W$.
It also follows that  
% $\thbar_1$ is determined by the requirement \eqref{ptnorm}, which yields
% $\thb%ar_1 = -\lambda \Lambda( -1 / \lambda).$ The value $W$ achieved for
% the decisio%n criterion \eqref{DP} is 
%$$W=\lambda(\thbar_1-1) = -\lambda \Lambda( -\frac{1}{\lambda})-\lambda.$$
$$W=\lambda \thbar_1 = -\lambda \Lambda( -\frac{1}{\lambda}).$$

\newpage
\bibliographystyle{abbrvnat}
% \bibliography{/Users/thomas/Literatur-Finance/central}
\bibliography{GenEnt}
\end{document}